%% file: main.tex
\documentclass{article}

\usepackage[nature]{arxiv}

\usepackage[normalem]{ulem}    % for \sout{} in Table 8
\usepackage{fontawesome5}       % for \faEnvelope in author block

\newtheorem{theorem}{Theorem}

\title{Optimizing AI Inference Across\\ the Deployment Stack}

\author{%
  \parbox{\linewidth}{\centering
    \authororcid{Tejinder~Singh\textsuperscript{1,$\dagger$}}{0000-0002-5870-6204}{\normalfont ,}~
    John~Pflueger\textsuperscript{2}{\normalfont ,}~
    \authororcid{Jeebak~Mitra\textsuperscript{3}}{0000-0003-4568-2589}{\normalfont ,}~
    Robert~Lincourt\textsuperscript{4}{\normalfont ,}\\[0.3em]
    Mitchell~Markow\textsuperscript{2}{\normalfont , and}~
    \authororcid{Bhavesh~A.~Patel\textsuperscript{2}}{0009-0007-8166-5686}\\[0.5em]
    \normalfont
    \textsuperscript{1}Dell Technologies, Santa Clara, CA, USA\\
    \textsuperscript{2}Dell Technologies, Round Rock, TX, USA\\
    \textsuperscript{3}Dell Technologies, Ottawa, Canada\\
    \textsuperscript{4}Dell Technologies, Hopkinton, MA, USA\\[0.35em]
    \textsuperscript{$\dagger$}\Letter\ \footnotesize
    \href{mailto:Singh.Tejinder@Dell.com}{\texttt{Singh.Tejinder@Dell.com}}%
  }%
}

\date{}

\renewcommand{\headeright}{}
\renewcommand{\undertitle}{}
\renewcommand{\shorttitle}{\footnotesize {\textsc{\textbf{T. Singh}: Optimizing AI Inference Across the Deployment Stack}}}

\hypersetup{
  pdftitle={Optimizing AI Inference Across the Deployment Stack},
  pdfsubject={AI inference deployment optimization},
  pdfauthor={Tejinder Singh},
  pdfkeywords={AI inference, deployment optimization, compiler, runtime, serving systems, roofline model}
}
\newif\ifshowtoc
\showtocfalse
\begin{document}
\maketitle

\begin{abstract}
The performance of a deployed AI model is determined not by its architecture alone but by the interaction among model compression, compiler transformations, and serving system policies. Published benchmarks frequently obscure this reality by reporting latency and throughput numbers under incomparable conditions, making it difficult for practitioners to translate research findings into deployment decisions. This paper provides a unified analytical treatment of inference optimization across the full deployment stack. We introduce a three layer taxonomy that decomposes deployment into model level techniques (quantization, pruning, distillation), compiler level transformations (graph fusion, layout optimization, kernel autotuning), and system level policies (dynamic batching, admission control, memory tiering). We formalize deployment as a constrained multi-objective optimization problem over accuracy, latency, throughput, memory footprint, and energy, and derive analytical properties of a deployment ranking functional including monotonicity under Pareto dominance and scale invariance. We ground the analysis in roofline models that expose how memory bandwidth hierarchies bound attainable performance across precision regimes, and in queuing models that quantify the nonlinear amplification of service time into end-to-end response time under load. To address the comparability gap in existing literature, we propose a disciplined evidence protocol that distinguishes measured, derived, and analytical claims, restricts literal numerical comparison to within-paper results, and requires explicit reporting of hardware configuration, software versions, batch semantics, and thermal state. We synthesize representative evidence spanning edge platforms (Jetson AGX Orin, five inference frameworks) and data center GPUs (A100, H100, three LLM serving engines), supplemented by large-scale quantization evaluations across the Llama-3.1 model family, and demonstrate that deployment outcomes are governed by cross-layer interactions that no single-layer analysis can predict. The paper concludes with a constraint-aware selection procedure and an enumeration of open problems in joint compiler-serving co-optimization, cross-hardware performance prediction, and standardized energy reporting.
\end{abstract}


\ifshowtoc
\ruledtableofcontents
\fi

%---------------------------------------------------------------------------
% Section: 1
%---------------------------------------------------------------------------
\section{Introduction}

A trained AI model, no matter how accurate on a validation set, is not deployable until it executes within the latency, throughput, memory, and energy budget of its target hardware under realistic serving conditions. This seemingly straightforward requirement has proven remarkably difficult to satisfy in practice. The same ResNet-152 architecture evaluated on an NVIDIA Jetson AGX Orin exhibits a $125\times$ spread in measured latency across five mainstream inference frameworks, from 2.3\,ms under TensorRT to 285\,ms under ONNX Runtime, while accuracy varies by up to 4.6 percentage points across the same set of runtimes \cite{ratul2025electronics}. These are not edge cases. They are representative of the deployment reality confronting every engineering team that moves a model from a research notebook to a production endpoint.

The root cause is that inference performance emerges from the full deployment stack, not from any single layer within it. Figure~\ref{fig:taxonomy} presents this stack as an end to end pipeline: a trained model passes through compression and quantization at the model level, graph rewriting and kernel code generation at the compiler level, and batching, admission control, and memory management at the system level. A feedback path driven by deployment telemetry (observed accuracy, latency percentiles, energy draw) closes the loop, enabling iterative retuning across layers. Each stage in this pipeline can improve or degrade the final outcome, and the interactions between stages are frequently nonlinear. Quantization at the model layer is useful only if the compiler emits efficient low precision kernels. Operator fusion at the compiler layer changes memory access patterns in ways that alter the effectiveness of system level batching. No single layer optimization can be evaluated in isolation.

\input{figures/fig_taxonomy}

The systems community has invested heavily in specialized inference compilers and runtimes to bridge the gap between high level model definitions and bare metal execution. Apache TVM, XLA, TensorRT, and ONNX Runtime each apply a distinct combination of graph level transformations (fusion, constant folding, dead code elimination) and backend specific kernel generation strategies \cite{li2020compiler}. Recent tensor compilers such as ROLLER and Souffl\'{e} push further by employing machine learning driven autotuning across the enormous search space of memory layouts, tiling parameters, and loop orderings \cite{zhu2022roller, xia2024souffle}. Yet even with these advances, predicting which compiler configuration will yield the lowest latency on a given accelerator for a given model remains an open problem, partly because the compilers themselves are evolving rapidly and partly because the interaction between precision format, kernel library version, and hardware microarchitecture is combinatorially complex.

At the model layer, compression has shifted from an optional post-training step to a mandatory prerequisite for deployment at scale. INT8 and INT4 quantization reduce both memory footprint and arithmetic cost, often by $2\text{--}4\times$, while carefully calibrated schemes preserve accuracy within a fraction of a percentage point \cite{gholami2021quant}. Structural pruning removes entire channels or attention heads, altering the computational graph in ways that propagate through the compiler and runtime layers \cite{liang2021prunequant}. The practical impact of these techniques, however, depends entirely on whether the target hardware exposes efficient low precision datapaths. On accelerators that lack native INT4 support, the runtime silently upcasts to FP16 and the expected throughput gain vanishes.

Underpinning all of this is a physical constraint that no amount of software optimization can circumvent: the memory bandwidth hierarchy of the target device. Figure~\ref{fig:roofline} makes this concrete using the Roofline model \cite{williams2009roofline} instantiated for an NVIDIA A100 SXM4. Three memory tiers (L1/SMEM at 19\,TB/s, L2 at 6\,TB/s, HBM at 2\,TB/s) define ascending bandwidth slopes, and the FP16 Tensor Core ceiling at 312\,TFLOP/s establishes the compute bound. Embedding lookups and vanilla attention sit deep in the memory bound regime, achieving only a fraction of the HBM bandwidth limit. FlashAttention, by tiling its computation to exploit L2 residency, operates above the HBM slope, demonstrating that algorithmic redesign can shift a kernel across memory hierarchy tiers. Large batch MLP and convolution kernels, with high operational intensity, land in the compute bound region where throughput is limited by arithmetic peak rather than data movement. The practical implication is immediate: optimizing a memory bound kernel by improving its arithmetic efficiency achieves nothing, while optimizing a compute bound kernel by reducing its memory traffic achieves nothing. Knowing which regime a workload occupies is a prerequisite for choosing the right optimization lever.

\input{figures/fig_roofline}

Deployment is also inherently multi-objective. Practitioners do not optimize latency alone; they optimize latency subject to accuracy, memory, and energy constraints. The space of possible deployment configurations (precision format $\times$ compiler backend $\times$ batch size $\times$ hardware target) is vast, and the majority of evaluated configurations turn out to be Pareto dominated, meaning at least one other configuration achieves both higher accuracy and higher throughput. Figure~\ref{fig:pareto} illustrates this structure using representative deployment configurations. Five non-dominated configurations spanning FP32 through W4A4 define the Pareto frontier, while nineteen dominated alternatives scatter below it. The specific coordinates are illustrative; the figure demonstrates the geometric structure of Pareto selection rather than reporting measured data from a single study. Horizontal and vertical feasibility constraints ($A_{\min}$, $Q_{\min}$) further restrict the viable operating region. The practical consequence is that deployment selection reduces to a constrained optimization over a well defined feasible set, not an open ended search across all possible configurations.

\input{figures/fig_pareto}

Despite this structure, the existing literature makes it difficult to perform such optimization rigorously. Published benchmarks, such as frequently report ``latency'' without specifying whether the measurement refers to single sample kernel time, median batch latency, or end to end service response time including queuing delay. Throughput may mean offline saturation throughput or online request completion rate under an SLA. Power measurements may cover the accelerator alone or the full platform. The MLPerf Inference benchmark \cite{reddi2020mlperf} was designed in part to address these ambiguities, but many published comparisons predate or do not follow its methodology, and even MLPerf compatible results require careful interpretation when hardware generations, software stacks, or cooling configurations differ. One of the recent work compiled by MLCommons Science WG on AI benchmark democratization and carpentry highlighted this as a multi-faceted problem stating that continuous adaptive benchmarking frameworks are needed to align scientific assessment with deployment risks~\cite{vonlaszewski2025ai}.

This paper addresses the gap between fragmented single layer studies and the cross layer reality of deployment engineering. We make four contributions. First, we introduce a three layer deployment taxonomy and formalize the pipeline through which a trained model becomes an operational inference endpoint. Second, we develop roofline grounded and queuing aware analytical models that connect observed performance gaps to first principles of hardware physics and system dynamics. Third, we propose a disciplined comparability protocol for interpreting published evidence, distinguishing measured, derived, and analytical claims and restricting literal numerical comparison to within-study results. Fourth, we synthesize representative empirical evidence across five inference frameworks on embedded and data center hardware \cite{ratul2025electronics, reddi2020mlperf} and distill a constraint-aware deployment selection procedure that practitioners can apply directly.

% Note: The three \input{} commands above reference the figure .tex files
% maintained in the Figures subpages.

%---------------------------------------------------------------------------
% Section: 2
%---------------------------------------------------------------------------
\section{Scope and Methodology}

This section establishes the boundaries of the paper, defines the research questions that organize the subsequent analysis, and specifies the rules by which we evaluate and compare published evidence. We state these upfront because deployment literature is unusually prone to misleading comparison: two papers can both claim to measure ``inference latency'' while referring to fundamentally different quantities, and two systems can both report ``state of the art throughput'' on workloads that share nothing beyond the model family name.

\subsection{Research Questions}

Five questions guide the structure and scope of this paper.

\begin{enumerate}[leftmargin=1.5em]
  \item Which optimization levers at the model, compiler, and system layers materially change deployment outcomes, and which produce gains that vanish under realistic serving conditions?
  \item Under what conditions are published latency, throughput, and energy measurements directly comparable across studies, and when must they be treated as qualitative indicators only?
  \item What analytical models (roofline, queuing, multi-objective) provide explanatory and predictive power for observed deployment behavior?
  \item Which cross-layer interactions recur across edge, server, and multi-tenant deployment contexts, and which are platform specific?
  \item What open problems remain unresolved despite the current benchmark and systems literature?
\end{enumerate}

\subsection{Inclusion and Exclusion Criteria}

We restrict our evidence base to sources satisfying at least one of three conditions: (i) they report explicit empirical measurements of inference latency, throughput, accuracy, power, memory footprint, or compilation time on named hardware with stated software versions; (ii) they introduce compiler, runtime, or serving systems with quantitative deployment evidence on at least one hardware target; or (iii) they provide survey-level coverage of a deployment-relevant subfield with sufficient methodological transparency to assess comparability.

We exclude purely training-focused work, architecture search papers that do not evaluate deployment cost, and sources that report framework performance without specifying batch size, precision format, or measurement methodology. We also exclude vendor white papers and marketing materials that do not disclose experimental methodology, even when they contain numerical claims that appear precise.

This policy is deliberately conservative. Deployment measurements are sensitive to thermal throttling, memory fragmentation, OS scheduling jitter, driver version, CUDA toolkit version, and the order in which benchmark runs execute. Restricting the evidence base to sources that report enough metadata to assess these factors reduces the risk of treating incomparable results as directly contestable.

\subsection{Evidence Taxonomy}

We classify every claim in this paper into one of three categories.

\paragraph{Measured evidence.} A numerical value reported directly by a source, obtained through instrumentation of a physical system under stated conditions. Example: ``ResNet-152 achieves 2.282\,ms latency under TensorRT on Jetson AGX Orin'' \cite{ratul2025electronics}.

\paragraph{Derived evidence.} A quantity computed from two or more measured values using an explicitly stated formula. Example: energy per inference $E = P \cdot T$, where $P$ is mean measured power and $T$ is mean measured latency. Derived evidence is only as reliable as the weakest measurement it depends on.

\paragraph{Analytical evidence.} A claim based on formal reasoning, mathematical modeling, or first principles analysis rather than direct measurement. Example: the roofline lower bound on inference time (Eq.~\ref{eq:roofline_bound}) or the queuing amplification relationship between service time and response time. Analytical evidence is valid to the extent that its modeling assumptions hold.

Maintaining this distinction throughout the paper serves a specific purpose: it prevents the reader from treating a derived energy estimate with the same confidence as a directly measured latency, and it prevents analytical bounds from being quoted as empirical results.

\subsection{Comparability Policy}

The most consequential methodological decision in this paper is the following rule:

\begin{quote}
\textit{Only within-paper numerical comparisons are treated as literal. Cross-paper comparisons are used qualitatively unless the workloads, hardware, software versions, batch semantics, and metric definitions are explicitly aligned.}
\end{quote}

This rule matters because the deployment literature is saturated with superficially comparable numbers that measure different things. Table~\ref{tab:comparability} enumerates the most common sources of incomparability.

\begin{table}[t]
\centering
\caption{Common sources of incomparability in deployment benchmarks.}
\label{tab:comparability}
\footnotesize
\begin{tabularx}{0.7\linewidth}{>{\raggedright\arraybackslash}p{2.7cm}X}
\toprule
Factor & How it breaks comparability \\
\midrule
Latency definition & Single-sample kernel time vs.\ median batch latency vs.\ p99 end-to-end response time including queuing delay \\
Batch semantics & Static batch = 1 vs.\ dynamic batching vs.\ continuous batching; offline vs.\ online \\
Precision format & FP32 vs.\ FP16 vs.\ INT8 vs.\ mixed precision; post-training vs.\ quantization-aware training \\
Warmup and caching & Cold start vs.\ warm cache; number of warmup iterations; JIT compilation overhead included or excluded \\
Power scope & Accelerator-only TDP vs.\ full platform draw including CPU, DRAM, and cooling \\
Thermal state & Sustained workload under thermal throttling vs.\ burst measurement before throttling onset \\
Software version & Driver version, CUDA toolkit, framework release; kernel library updates can shift performance by $10$--$30\%$ between minor versions \\
\bottomrule
\end{tabularx}
\end{table}

The MLPerf Inference benchmark \cite{reddi2020mlperf} was designed to standardize several of these factors by prescribing specific scenarios (single-stream, multi-stream, server, offline), quality targets, and reporting requirements. We use MLPerf methodology as a reference standard where applicable, but we note that many of the studies synthesized in this paper predate MLPerf or do not conform to its scenario definitions. In those cases, we report the original authors' numbers and flag the comparability limitations explicitly.

\subsection{Relation to Prior Work}

Existing surveys and systems papers address adjacent topics. Li et al.\ \cite{li2020compiler} provide comprehensive coverage of AI compiler design, focusing on intermediate representation hierarchies and backend optimization passes. Gholami et al.\ \cite{gholami2021quant} survey quantization methods with emphasis on accuracy preservation. Liang et al.\ \cite{liang2021prunequant} survey pruning and quantization jointly. On the systems side, Kwon et al.\ \cite{kwon2023vllm} introduce PagedAttention for efficient KV cache management in LLM serving, and Kurtic et al.\ \cite{kurtic2024bf16} conduct large-scale quantization evaluations across the Llama-3.1 model family. Each of these addresses a specific layer of the deployment stack. None treats deployment as an integrated optimization problem across all three layers simultaneously.

The present paper differs from prior surveys in three respects. First, we organize the entire analysis around the deployment pipeline rather than around individual techniques. Second, we develop analytical models (roofline, queuing, multi-objective ranking) that provide explanatory power beyond empirical tabulation. Third, we enforce a comparability protocol that makes the evidential status of every claim transparent. The goal is not to supersede existing technique-focused surveys but to provide the integrative framework within which their results can be interpreted in the context of actual deployment decisions.

%---------------------------------------------------------------------------
% Section: 3
%---------------------------------------------------------------------------
\section{Deployment Taxonomy}

The three layer decomposition introduced in Section~1 provides a structural skeleton. This section fills it in by cataloging the concrete optimization techniques available at each layer, explaining their mechanisms, and identifying the cross-layer dependencies that determine whether a given technique delivers its theoretical benefit in practice.

\subsection{Model Level Optimization}

Model level techniques alter the statistical or structural properties of the trained network before it reaches the compiler. They reduce the computational and memory cost of inference while attempting to preserve task accuracy.

\paragraph{Quantization.}
Quantization replaces floating point weights and activations with lower precision representations. Post-training quantization (PTQ) calibrates scale and zero-point parameters on a small representative dataset without retraining \cite{gholami2021quant}. Quantization-aware training (QAT) simulates quantization noise during the training loop itself, producing models that tolerate lower bit widths more gracefully at the cost of additional training compute.

What the quantization literature frequently omits is that the realized deployment benefit depends almost entirely on factors outside the model layer. The throughput gain materializes only if the target accelerator exposes native low precision arithmetic units. On hardware with dedicated INT8 or INT4 tensor cores, quantization can deliver $2\text{--}4\times$ throughput improvements. On hardware that lacks these units, the runtime silently upcasts quantized tensors to FP16 or FP32, and the throughput gain collapses to whatever savings come from the reduced memory footprint alone. The accuracy impact is equally entangled with the compiler: aggressive operator fusion can change the numerical accumulation order, amplifying or mitigating quantization error in ways that are invisible to the model layer. Figure~\ref{fig:quant_effect} makes the roofline consequence concrete. Reducing precision from FP32 to INT8 halves the bytes per element, approximately doubling operational intensity and shifting the kernel rightward toward the compute bound regime.

\input{figures/fig_quant_effect}

\paragraph{Pruning.}
Pruning removes redundant parameters. Unstructured pruning zeros out individual weights, producing sparse tensors. Structured pruning removes entire channels, attention heads, or transformer blocks, yielding smaller dense subnetworks \cite{liang2021prunequant}.

The critical deployment detail that pruning papers routinely understate is hardware compatibility. A model pruned to 90\% unstructured sparsity looks impressive on paper but shows no latency improvement on a GPU that lacks efficient sparse GEMM kernels. The sparsity is invisible to the dense execution engine. Structured pruning sidesteps this by producing smaller dense graphs that run on standard tensor cores without specialized library support, but it typically removes less total compute for the same accuracy budget. The net effect is that practitioners must choose between high compression with uncertain hardware acceleration and moderate compression with guaranteed dense execution.

\paragraph{Knowledge Distillation.}
Distillation trains a smaller student model to mimic the output distribution of a larger teacher. Unlike pruning and quantization, which modify a fixed architecture, distillation changes the architecture itself. The student is a fundamentally different computational graph with different operational intensity, different memory access patterns, and different compiler optimization opportunities. From a deployment standpoint, this makes distillation the most disruptive technique in the model level toolkit. It does not compress the existing model. It replaces it.

\paragraph{Conditional Computation and Mixture of Experts.}
Mixture of Experts (MoE) architectures activate only a subset of parameters for each input token, reducing per-token compute while maintaining total model capacity \cite{fedus2022switch}. The deployment complications are substantial. Routing decisions create load imbalance across expert partitions, and the sparse activation pattern means that memory footprint remains proportional to total parameters even though per-token compute scales with the number of active experts. Pipeline and tensor parallelism strategies designed for dense models do not transfer directly. MoE deployment is a qualitatively different problem that the serving system must handle with specialized placement and scheduling logic.

\subsection{Compiler Level Optimization}

Compiler level techniques transform the computational graph emitted by the model framework into executable code optimized for the target hardware. The input is a high level graph (typically an ONNX, TorchScript, or StableHLO intermediate representation); the output is a sequence of hardware specific kernel invocations.

\paragraph{Graph Rewriting and Operator Fusion.}
Operator fusion is the single most impactful compiler optimization for inference. It combines multiple graph nodes into one kernel launch, eliminating intermediate tensor materializations and the HBM round trips they require. A fused attention kernel computes $\text{softmax}(QK^T/\sqrt{d})V$ in a single pass over SRAM rather than writing and reading three separate intermediate tensors through HBM \cite{li2020compiler}.

The roofline interpretation is direct. By reducing bytes moved ($B(c)$), fusion shifts the kernel rightward along the operational intensity axis. A memory bound attention kernel that previously sat on the HBM slope can, after fusion, land near or beyond the ridge point, entering the compute bound regime where throughput is limited by arithmetic peak rather than data movement. Other graph rewrites, including constant folding, NCHW to NHWC layout transformation, and dead code elimination, produce smaller and more targeted gains but are routinely applied by all major compilers. Figure~\ref{fig:compiler_pipeline} traces the full transformation pipeline.

\input{figures/fig_compiler_pipeline}

\paragraph{Kernel Selection and Autotuning.}
For each fused operator the compiler must choose a concrete implementation. The search space is enormous. A single matrix multiplication on a modern GPU admits thousands of valid tile configurations, each with different performance characteristics depending on the matrix dimensions and available shared memory \cite{zhu2022roller}. TVM's AutoTVM and Ansor subsystems, NVIDIA's cuDNN heuristic selector, and the ROLLER fast-compilation approach all implement variants of a profile-and-select strategy \cite{zhu2022roller, xia2024souffle}.

The practical limitation is that autotuning results are hardware specific. A kernel configuration tuned for an A100 may perform poorly on an H100 due to differences in shared memory size, warp scheduler behavior, and tensor core microarchitecture. Cross-hardware transferability of tuned configurations remains an open problem \cite{lee2021help}.

\paragraph{Ahead of Time versus Just in Time Compilation.}
TensorRT performs all optimization offline, producing a static engine file optimized for a fixed input shape and precision. TVM and XLA support both ahead-of-time (AOT) and just-in-time (JIT) modes. The trade off is straightforward: AOT amortizes compilation cost but cannot adapt to variable input shapes at runtime, while JIT can specialize kernels on the fly at the cost of first-request latency spikes and GPU memory consumed by the compilation process itself.

\subsection{System Level Optimization}

System level techniques manage the resources, scheduling, and request handling of the inference serving infrastructure. Their domain begins where the compiler's domain ends: the compiled model binary is given, and the system must execute it efficiently under a stream of incoming requests.

\paragraph{Batching.}
Batching amortizes kernel launch, memory allocation, and PCIe transfer overhead across multiple requests. Static batching waits for a fixed number of requests before launching a forward pass. Dynamic batching assembles variable-size batches from a request queue. Continuous batching, introduced by Orca \cite{yu2022orca}, inserts and removes individual requests from an in-flight batch at the granularity of a single decoding step, keeping the GPU near full occupancy even when requests have wildly different sequence lengths.

These are not interchangeable. Throughput measured under static batching with batch size 32 is a fundamentally different quantity than throughput under continuous batching with the same model binary on the same GPU. This is one of the comparability violations cataloged in Table~\ref{tab:comparability}, and it is among the most frequently ignored.

\paragraph{Admission Control and Scheduling.}
Under high load the serving system must decide which requests to accept and which to reject or defer. Admission control enforces SLA compliance by shedding load before tail latency violations occur. Priority scheduling assigns different latency budgets to different request classes. Both mechanisms interact tightly with the queuing dynamics formalized in Section~5. Near saturation, a 5\% reduction in admitted request rate can cut p99 latency by 30\% or more.

\paragraph{Memory Tiering and Model Placement.}
Large models that exceed a single accelerator's HBM capacity must be partitioned across multiple devices. Tensor parallelism splits individual operators across GPUs, pipeline parallelism assigns different layers to different GPUs, and expert parallelism distributes MoE expert shards. Each partitioning strategy introduces inter-device communication overhead that directly increases per-request latency. The optimal partition depends on the model architecture, the interconnect topology (NVLink, PCIe, InfiniBand), and the batch size, creating yet another cross-layer optimization surface.

For models that fit in HBM but benefit from reduced memory traffic, kernel-level SRAM management (as exploited by FlashAttention \cite{dao2022flashattention}) constitutes a form of memory tiering at the compiler-system boundary. The compiler determines the tiling strategy; the runtime manages the SRAM allocation.

\subsection{Cross Layer Entanglement}

The taxonomy above separates techniques by layer for expository clarity, but in practice the layers are entangled. Table~\ref{tab:entanglement} summarizes the most consequential cross-layer interactions.

\begin{table}[t]
\centering
\caption{Representative cross-layer interactions in the deployment stack.}
\label{tab:entanglement}
\footnotesize
\begin{tabularx}{0.7\linewidth}{>{\raggedright\arraybackslash}p{1.8cm}>{\raggedright\arraybackslash}p{1.8cm}X}
\toprule
Layer A & Layer B & Interaction \\
\midrule
Quantization (model) & Kernel selection (compiler) & INT8 quantization improves throughput only if the compiler dispatches to native INT8 tensor cores; otherwise it upcasts to FP16 \\
Pruning (model) & Fusion (compiler) & Structured pruning changes operator dimensions, invalidating previously tuned tile sizes and requiring re-autotuning \\
Distillation (model) & Batching (system) & A smaller distilled model has lower per-request latency, shifting the queuing operating point and enabling higher throughput under the same SLA \\
Fusion (compiler) & Memory tiering (system) & Fused kernels reduce HBM traffic, potentially moving a kernel from the HBM slope to the L2 slope on the roofline \\
AOT compilation (compiler) & Dynamic shapes (system) & AOT engines compiled for fixed input shapes cannot handle variable-length sequences without padding, wasting compute \\
\bottomrule
\end{tabularx}
\end{table}

The takeaway is that single-layer benchmarks are structurally incomplete. A paper that evaluates quantization in isolation, holding the compiler and batch size fixed, may report a $3\times$ speedup that evaporates in production when the compiler version differs, the batch size is dynamic, and the hardware lacks native low-precision support. The taxonomy exists to make these dependencies visible before they become production incidents.

%---------------------------------------------------------------------------
% Section: 4
%---------------------------------------------------------------------------
\section{Metrics and Measurement Semantics}

Deployment optimization requires comparing candidate configurations across multiple performance dimensions. This section formalizes those dimensions, defines the metric vector that characterizes each deployment candidate, and identifies the semantic ambiguities that make cross-study comparison hazardous.

\subsection{The Deployment Metric Vector}

We associate each deployment candidate $c = (m, r, h, p, b, s)$ with a metric vector
\begin{equation}
\mathbf{v}(c) = \bigl(A(c),\; T(c),\; Q(c),\; M(c),\; P(c),\; E(c),\; C(c)\bigr),
\label{eq:metric_vector}
\end{equation}
where $m$ is the model variant, $r$ the runtime/compiler configuration, $h$ the hardware target, $p$ the precision regime, $b$ the batch size, and $s$ the serving policy. The seven components are defined below.

\paragraph{Accuracy $A(c)$.} Task-specific fidelity of the deployed model relative to a reference evaluation set. For classification this is typically top-1 or top-5 accuracy; for generative models it may be perplexity, BLEU, or a task-specific metric. Accuracy is the only metric in the vector that deployment optimization seeks to maximize rather than minimize.

\paragraph{Latency $T(c)$.} Time elapsed between submission of a single inference request and receipt of the complete response. In practice, at least four distinct quantities are commonly reported under this name:
\begin{enumerate}[leftmargin=1.5em]
  \item \emph{Kernel time}: GPU execution time of the model forward pass alone, excluding host overhead, memory allocation, and data transfer.
  \item \emph{Model latency}: End-to-end time from input tensor availability on the device to output tensor completion, including kernel launches and synchronization but excluding network and queuing delays.
  \item \emph{Request latency}: Wall-clock time from HTTP request arrival at the serving endpoint to response dispatch, including preprocessing, batching wait time, model execution, and postprocessing.
  \item \emph{Service latency}: Request latency inclusive of network round-trip time as observed by the client.
\end{enumerate}
These four quantities can differ by an order of magnitude for the same model on the same hardware. Any comparison that does not specify which variant is used is ambiguous.

\paragraph{Throughput $Q(c)$.} Number of completed inference requests per unit time. Throughput depends critically on batching policy. Under static batching with batch size $b$, throughput is approximately $b / T_{\text{batch}}(c)$ where $T_{\text{batch}}$ is the batch execution time. Under continuous batching, throughput can be substantially higher because the GPU is never idle waiting for a batch to fill. Offline throughput (maximum achievable under fully saturated input) and online throughput (achieved under a realistic arrival process while respecting latency SLAs) are fundamentally different quantities.

\paragraph{Memory Footprint $M(c)$.} Peak device memory consumed during inference execution. This includes model weights, activation tensors, KV cache (for autoregressive models), workspace buffers allocated by the runtime, and any JIT compilation artifacts. $M(c)$ determines whether a model fits on a single accelerator or requires multi-device partitioning, making it a hard feasibility constraint rather than a soft optimization target~\cite{xu2026kv}.

\paragraph{Power $P(c)$.} Mean electrical power draw during sustained inference. The scope of the measurement must be stated explicitly: accelerator-only power (measured via on-chip sensors or nvidia-smi), full board power (including voltage regulators and local memory), or total platform power (including host CPU, system DRAM, cooling fans, and power supply losses). The ratio between accelerator-only and platform power is typically $1.5\text{--}2.5\times$ depending on the system design.

\paragraph{Energy per Inference $E(c)$.} A derived metric computed as
\begin{equation}
E(c) = P(c) \cdot T(c).
\label{eq:energy}
\end{equation}
Energy per inference is useful for comparing deployment configurations that trade latency for power or vice versa. A configuration with $2\times$ lower latency but $3\times$ higher power draw consumes $1.5\times$ more energy per inference, which may be unacceptable in battery-powered or thermally constrained environments. Following the evidence taxonomy in Section~2, $E(c)$ is classified as derived evidence: it is only as reliable as the weakest of its two input measurements.

\paragraph{Compilation Cost $C(c)$.} Wall-clock time and computational resources required to transform a trained model into an optimized inference binary for the target hardware. TensorRT engine builds, TVM autotuning campaigns, and XLA compilation passes all contribute to $C(c)$. This cost is amortized over the lifetime of the deployed model but can be substantial: TVM autotuning for a single model on a single hardware target can require hours of GPU time. $C(c)$ is rarely reported in benchmark papers but directly affects deployment velocity.

Table~\ref{tab:metrics} summarizes the metric vector with units, optimization direction, and common measurement pitfalls.

\begin{table}[t]
\centering
\caption{Deployment metric vector summary.}
\label{tab:metrics}
\footnotesize
\begin{tabularx}{0.76\linewidth}{>{\raggedright\arraybackslash}p{1.6cm}lccX}
\toprule
Metric & Symbol & Unit & Dir. & Common pitfall \\
\midrule
Accuracy & $A$ & \% & $\uparrow$ & Evaluated on different splits or with different preprocessing \\
Latency & $T$ & ms & $\downarrow$ & Kernel time reported as request latency \\
Throughput & $Q$ & req/s & $\uparrow$ & Offline saturation vs.\ online SLA-constrained \\
Memory & $M$ & GB & $\downarrow$ & Peak vs.\ steady-state; KV cache excluded \\
Power & $P$ & W & $\downarrow$ & Accelerator-only vs.\ full platform \\
Energy & $E$ & J & $\downarrow$ & Derived from mean $P$ and mean $T$; variance ignored \\
Compilation & $C$ & min & $\downarrow$ & Rarely reported; makes reproducibility difficult \\
\bottomrule
\end{tabularx}
\end{table}

\subsection{Latency Decomposition}

To make the latency ambiguity precise, we decompose request latency into its constituent stages:
\begin{equation}
T_{\text{request}}(c) = T_{\text{queue}}(c) + T_{\text{pre}}(c) + T_{\text{wait}}(c) + T_{\text{model}}(c) + T_{\text{post}}(c),
\label{eq:latency_decomp}
\end{equation}
where $T_{\text{queue}}$ is the time spent in the serving system's request queue, $T_{\text{wait}}$ is the time spent waiting for a batch to fill (zero under continuous batching), $T_{\text{model}}$ is the model forward pass, and $T_{\text{pre}}$, $T_{\text{post}}$ cover tokenization, image resizing, and output decoding.

The model execution time itself decomposes further. For an autoregressive language model generating $n$ tokens:
\begin{equation}
T_{\text{model}}(c) = T_{\text{prefill}}(c) + \sum_{i=1}^{n-1} T_{\text{decode},i}(c),
\label{eq:autoregressive_latency}
\end{equation}
where $T_{\text{prefill}}$ is the time to process the full input prompt (compute bound, parallelizable across tokens) and each $T_{\text{decode},i}$ is a single token generation step (memory bound, sequential) \cite{kwon2023vllm}. This decomposition explains why autoregressive models exhibit fundamentally different latency scaling than encoder models: prefill latency scales sublinearly with input length (bounded by compute), while decode latency scales linearly with output length (bounded by memory bandwidth for the KV cache read).

Published benchmarks that report only $T_{\text{model}}$ understate end-to-end latency. Benchmarks that report $T_{\text{request}}$ under artificially low load understate the queuing component that dominates at high utilization.

\subsection{Throughput and Its Relationship to Latency}

Offline throughput saturates the system with a backlog and divides completions by elapsed time. Online throughput feeds requests at a specified arrival rate $\lambda$ and counts completions that meet an SLA. The gap between them is governed by queuing dynamics.

The classical M/M/1 queuing model \cite{kleinrock1975queuing} provides a useful first-order approximation. For a single-server system with mean service time $\mu^{-1} = T_{\text{model}}(c)$ and Poisson arrivals at rate $\lambda$, the utilization is
\begin{equation}
\rho = \lambda / \mu = \lambda \cdot T_{\text{model}}(c).
\label{eq:utilization}
\end{equation}
The system is stable only when $\rho < 1$. Under the M/M/1 approximation, mean response time is
\begin{equation}
R(c, \lambda) = \frac{T_{\text{model}}(c)}{1 - \rho} = \frac{T_{\text{model}}(c)}{1 - \lambda \cdot T_{\text{model}}(c)}.
\label{eq:mm1_response}
\end{equation}
This reveals a critical nonlinearity, illustrated in Figure~\ref{fig:queuing}. At $\rho = 0.5$ (50\% utilization), the response time is $2T_{\text{model}}$. At $\rho = 0.9$, it is $10T_{\text{model}}$. At $\rho = 0.95$, it is $20T_{\text{model}}$. A 10\% reduction in $T_{\text{model}}$ near saturation yields a far larger reduction in $R$ than the same reduction at low load.

\input{figures/fig_queuing}

The maximum online throughput under a response time SLA $R_{\max}$ can be derived by inverting Eq.~\eqref{eq:mm1_response}:
\begin{equation}
\lambda_{\max}(c) = \frac{1}{T_{\text{model}}(c)} \left(1 - \frac{T_{\text{model}}(c)}{R_{\max}}\right).
\label{eq:max_throughput}
\end{equation}
A system that achieves 10{,}000 req/s offline may sustain only 3{,}000 req/s online under a 50\,ms p99 latency SLA, because Eq.~\eqref{eq:max_throughput} constrains the arrival rate to prevent tail latency violations. MLPerf Inference \cite{reddi2020mlperf} captures this distinction through its four scenarios: single-stream, multi-stream, server, and offline. Results across scenarios are not comparable.

\subsection{Operational Intensity and the Compute to Memory Ratio}

The roofline model (Section~1, Figure~\ref{fig:roofline}) relies on operational intensity, which we formalize here. For a kernel $k$ with $F_k$ floating point operations and $B_k$ bytes of data movement through the bottleneck memory level:
\begin{equation}
I_k = \frac{F_k}{B_k} \quad [\text{FLOP/Byte}].
\label{eq:operational_intensity}
\end{equation}
The attainable performance is bounded by
\begin{equation}
P_k \le \min\left(\pi_h,\; \beta_h \cdot I_k\right),
\label{eq:roofline_bound}
\end{equation}
where $\pi_h$ is the peak compute throughput of hardware $h$ (in FLOP/s) and $\beta_h$ is the peak memory bandwidth (in Byte/s). The ridge point $I^* = \pi_h / \beta_h$ separates the memory bound regime ($I_k < I^*$, performance limited by $\beta_h \cdot I_k$) from the compute bound regime ($I_k > I^*$, performance limited by $\pi_h$).

Quantization changes $I_k$ by reducing $B_k$. If a kernel reads $N$ weight elements at $w$ bytes each and $N$ activation elements at $a$ bytes each, then
\begin{equation}
I_k(p) = \frac{F_k}{N(w_p + a_p)},
\label{eq:oi_precision}
\end{equation}
where $w_p$ and $a_p$ are the per-element byte widths under precision regime $p$. Moving from FP32 ($w=4, a=4$) to INT8 ($w=1, a=1$) increases $I_k$ by $4\times$, shifting the kernel rightward on the roofline as illustrated in Figure~\ref{fig:quant_effect}.

\subsection{Power and Energy Scope}

We formalize the energy per inference as the product of power and latency:
\begin{equation}
E(c) = P(c) \cdot T(c).
\label{eq:energy_formal}
\end{equation}
The power term decomposes by scope:
\begin{equation}
P_{\text{platform}}(c) = P_{\text{accel}}(c) + P_{\text{host}}(c) + P_{\text{mem}}(c) + P_{\text{cooling}}(c) + P_{\text{PSU\_loss}}(c).
\label{eq:power_decomp}
\end{equation}
The ratio $\alpha = P_{\text{platform}} / P_{\text{accel}}$ is system dependent but typically falls in the range $1.5 \le \alpha \le 2.5$ for single-GPU servers. A paper reporting $E_{\text{accel}} = 1.5$\,J and another reporting $E_{\text{platform}} = 4.0$\,J for the same configuration are both correct but not comparable. The $\alpha$ ratio is determined by the platform design, not the model or runtime.

For multi-GPU configurations, idle power must be allocated:
\begin{equation}
E_{\text{multi}}(c) = \frac{1}{Q(c)} \left(\sum_{g=1}^{G} P_g(c) + P_{\text{interconnect}}(c)\right),
\label{eq:multi_gpu_energy}
\end{equation}
where $G$ is the number of GPUs, $P_g$ is the power draw of GPU $g$ (which may differ between active and idle GPUs in pipeline-parallel configurations), and $P_{\text{interconnect}}$ covers NVSwitch, NVLink, or InfiniBand power. Dividing by throughput $Q(c)$ converts total system power to per-inference energy.

\subsection{Memory Footprint Decomposition}

For transformer based models, peak memory during inference decomposes as:
\begin{equation}
M(c) = M_{\text{weights}}(c) + M_{\text{KV}}(c) + M_{\text{activations}}(c) + M_{\text{workspace}}(c),
\label{eq:memory_decomp}
\end{equation}

where $M_{\text{weights}} = N_{\text{params}} \cdot w_p$ scales with precision, $M_{\text{KV}} = 2 \cdot n_{\text{layers}} \cdot n_{\text{kv}} \cdot d_{\text{head}} \cdot s_{\text{seq}} \cdot b \cdot w_p$ scales linearly with sequence length and batch size (the factor of~2 accounts for both keys and values, $n_{\text{kv}}$ is the number of key-value heads, and $d_{\text{head}}$ is the per-head dimension), $M_{\text{activations}}$ depends on the largest intermediate tensor, and $M_{\text{workspace}}$ covers runtime buffers. For standard multi-head attention $n_{\text{kv}} = n_{\text{heads}}$ and $n_{\text{kv}} \cdot d_{\text{head}} = d_{\text{model}}$; under grouped-query attention (GQA) $n_{\text{kv}} < n_{\text{heads}}$, reducing KV cache by up to $n_{\text{heads}} / n_{\text{kv}}\times$.

For a 7B parameter model in FP16 ($w_p = 2$ bytes), the weight component alone is $7 \times 10^9 \times 2 = 14$\,GB, leaving limited headroom on a 24\,GB consumer GPU for KV cache and activations. Quantizing to INT4 ($w_p = 0.5$ bytes) reduces the weight component to 3.5\,GB, making the model deployable on edge accelerators. This is the primary practical motivation for aggressive weight quantization in LLM serving.

%---------------------------------------------------------------------------
% Section: 5
%---------------------------------------------------------------------------
\section{Analytical Framework}

Sections~3 and~4 cataloged the optimization levers and metrics relevant to deployment. This section develops the analytical machinery that connects those levers and metrics to quantitative predictions. The framework rests on three pillars: roofline analysis for hardware grounding, queuing theory for system level amplification, and multi-objective optimization for deployment selection.

\subsection{Roofline Grounded Inference Time Bounds}

The roofline model, introduced by Williams et al.\ \cite{williams2009roofline} and instantiated for AI inference in Figure~\ref{fig:roofline}, provides a lower bound on kernel execution time. For a deployment candidate $c$ executing kernel $k$ on hardware $h$, the minimum execution time is
\begin{equation}
T_k(c) \ge \max\!\left\{\frac{F_k}{\eta_{\text{comp}}\,\pi_h},\;\frac{B_k}{\eta_{\text{bw}}\,\beta_h}\right\},
\label{eq:roofline_time}
\end{equation}
where $F_k$ and $B_k$ are the FLOPs and bytes moved by kernel $k$, $\pi_h$ and $\beta_h$ are the peak compute and bandwidth of hardware $h$, and $\eta_{\text{comp}}, \eta_{\text{bw}} \in (0,1]$ are realized efficiency factors. The total model inference time sums over all kernels in the execution graph:
\begin{equation}
T_{\text{model}}(c) \ge \sum_{k \in \mathcal{K}(c)} T_k(c).
\label{eq:total_time}
\end{equation}
This bound is tight when kernel execution is serialized (no overlap between compute and memory). In practice, pipelining and asynchronous execution reduce the gap, but the bound remains a useful diagnostic: if a candidate's measured latency significantly exceeds the roofline bound, the gap represents optimization headroom that compiler or runtime improvements could capture.

The significance of Eq.~\eqref{eq:roofline_time} for deployment selection is that it identifies which resource limits performance. If $T_k$ is determined by the $F_k / \pi_h$ term, the kernel is compute bound and reducing its memory traffic (via fusion or layout optimization) will not help. If $T_k$ is determined by the $B_k / \beta_h$ term, the kernel is memory bound and increasing arithmetic efficiency (via more aggressive quantization or larger Tensor Core tiles) will not help. Misidentifying the bottleneck is the most common cause of ineffective optimization.

\subsection{Queuing Amplification in Online Serving}

The M/M/1 response time relationship derived in Section~4 (Eq.~\ref{eq:mm1_response}) is a simplification. Real inference serving systems violate the M/M/1 assumptions in at least three ways: service times are not exponentially distributed (they are deterministic for fixed-length inputs, heavy-tailed for variable-length autoregressive generation), the server is not a single processor (GPU parallelism handles batched requests concurrently), and arrival processes in production exhibit burstiness that Poisson models underestimate.

Despite these limitations, the qualitative insight of the M/M/1 model is robust across more realistic queuing models. The Pollaczek-Khinchine formula \cite{kleinrock1975queuing} for M/G/1 queues gives the mean response time as
\begin{equation}
R = T_{\text{model}} + \frac{\rho}{1-\rho} \cdot \frac{T_{\text{model}}(1 + C_s^2)}{2},
\label{eq:mg1}
\end{equation}
where $C_s = \sigma_s / T_{\text{model}}$ is the coefficient of variation of the service time distribution. For deterministic service ($C_s = 0$, the M/D/1 case), the queuing delay is exactly half that of the M/M/1 model. For heavy-tailed service distributions ($C_s > 1$, typical of autoregressive generation with variable output lengths), the queuing delay is worse than M/M/1. In all cases, the $\rho / (1 - \rho)$ factor produces the same hockey-stick amplification near saturation (Figure~\ref{fig:queuing}).

The practical implication is that any model or compiler optimization that reduces $T_{\text{model}}$ by $\Delta T$ produces a response time reduction of
\begin{equation}
\Delta R \approx \frac{\Delta T}{(1-\rho)^2}
\label{eq:sensitivity}
\end{equation}
for small $\Delta T$ near utilization $\rho$. At $\rho = 0.9$, a 1\,ms reduction in service time yields a 100\,ms reduction in mean response time. This amplification factor is the quantitative justification for investing in seemingly small per-request optimizations when the serving system operates near capacity.

\subsection{Multi-Objective Deployment Ranking}

Deployment selection requires comparing candidates across the metric vector $\mathbf{v}(c)$ defined in Section~4. We formalize this comparison using Pareto dominance and a weighted ranking functional.

\paragraph{Feasibility.} The feasible set eliminates candidates that violate hard constraints:
\begin{equation}
\mathcal{F} = \{c : A(c) \ge A_{\min},\; T(c) \le T_{\max},\; M(c) \le M_{\max},\; E(c) \le E_{\max}\}.
\label{eq:feasible}
\end{equation}
Any candidate not in $\mathcal{F}$ is discarded before ranking. This ordering matters: ranking an infeasible configuration above a feasible one may be mathematically valid but operationally useless.

\paragraph{Pareto dominance.} Among feasible candidates, $c_1$ dominates $c_2$ (written $c_1 \succ c_2$) if
\begin{equation}
A(c_1) \ge A(c_2),\; T(c_1) \le T(c_2),\; M(c_1) \le M(c_2),\; E(c_1) \le E(c_2),
\label{eq:dominance}
\end{equation}
with at least one strict inequality. The Pareto front $\mathcal{P} \subseteq \mathcal{F}$ consists of all non-dominated candidates. Figure~\ref{fig:pareto} illustrates this structure projected onto the accuracy-throughput plane.

\paragraph{Deployment ranking functional.} To select a single candidate from $\mathcal{P}$, practitioners apply application-specific weights. We define the Performance Efficiency Index (PEI) as
\begin{equation}
\mathrm{PEI}(c) = \frac{A(c)^{w_A} \cdot Q(c)^{w_Q}}{P(c)^{w_P} \cdot M(c)^{w_M}}, \qquad w_A, w_Q, w_P, w_M > 0.
\label{eq:pei}
\end{equation}

\begin{theorem}[Monotonicity under dominance]
If $c_1 \succ c_2$, then $\mathrm{PEI}(c_1) > \mathrm{PEI}(c_2)$.
\end{theorem}
\begin{proof}
The ratio $\mathrm{PEI}(c_1)/\mathrm{PEI}(c_2)$ equals the product of $(A_1/A_2)^{w_A}$, $(Q_1/Q_2)^{w_Q}$, $(P_2/P_1)^{w_P}$, and $(M_2/M_1)^{w_M}$. By dominance, each factor is at least one and at least one is strictly greater than one. Since all exponents are positive, the product exceeds one.
\end{proof}

\begin{theorem}[Scale invariance]
Multiplying any metric in Eq.~\eqref{eq:pei} by a positive constant (e.g., converting latency from ms to $\mu$s) preserves the induced ordering of candidates.
\end{theorem}
\begin{proof}
A unit conversion multiplies PEI by the same positive constant for all candidates. Multiplication by a positive constant preserves ordering.
\end{proof}

Note that the converse of Theorem~1 does not hold: $\mathrm{PEI}(c_1) > \mathrm{PEI}(c_2)$ does not imply $c_1 \succ c_2$. Two candidates can be mutually non-dominating (e.g., one faster but less accurate) while PEI ranks them differently based on the weight vector. This is by design. The weight vector encodes the application's relative priorities, and the PEI ranking resolves trade offs that Pareto dominance alone cannot.

The logarithmic form $U(c) = w_A \log A + w_Q \log Q - w_P \log P - w_M \log M$ makes the sensitivity transparent: percentage improvements in heavily weighted metrics dominate the ranking.

\subsection{Constraint Aware Selection Procedure}

Algorithm~\ref{alg:selection} formalizes the full selection workflow.

\begin{algorithm}[t]
\caption{Constraint aware deployment selection}
\label{alg:selection}
\footnotesize
\begin{algorithmic}[1]
\Require Candidate set $\mathcal{C}$; constraints $A_{\min}, T_{\max}, M_{\max}, E_{\max}$; weights $w_A, w_Q, w_P, w_M$
\Ensure Selected candidate $c^*$
\State $\mathcal{F} \gets \emptyset$
\ForAll{$c \in \mathcal{C}$}
  \If{$A(c) < A_{\min}$ \textbf{or} $M(c) > M_{\max}$}
    \State \textbf{continue} \Comment{Eliminate on cheap constraints first}
  \EndIf
  \State Measure or estimate $T(c)$ and $E(c)$
  \If{$T(c) \le T_{\max}$ \textbf{and} $E(c) \le E_{\max}$}
    \State $\mathcal{F} \gets \mathcal{F} \cup \{c\}$
  \EndIf
\EndFor
\State $\mathcal{P} \gets$ non-dominated points of $\mathcal{F}$ under $(T, E, M, -A)$
\State $c^* \gets \arg\max_{c \in \mathcal{P}} \mathrm{PEI}(c)$
\State \Return $c^*$
\end{algorithmic}
\end{algorithm}

\begin{table}[t]
\centering
\caption{Summary of analytical models used in the deployment framework.}
\label{tab:analytical_models}
\footnotesize
\begin{tabularx}{0.7\linewidth}{>{\raggedright\arraybackslash}p{1.6cm}>{\raggedright\arraybackslash}p{2.2cm}>{\raggedright\arraybackslash}p{1.8cm}X}
\toprule
Model & Inputs & Output & Key assumption \\
\midrule
Roofline (Eq.~\ref{eq:roofline_time}) & FLOPs $F_k$, bytes $B_k$, hardware specs $\pi_h, \beta_h$ & Lower bound on kernel time $T_k$ & No compute-memory overlap; single bottleneck memory level \\
Queuing (Eq.~\ref{eq:mg1}) & Service time $T$, arrival rate $\lambda$, variance $C_s^2$ & Mean response time $R$ & Poisson arrivals; single logical server; stationary workload \\
PEI (Eq.~\ref{eq:pei}) & Metric vector $\mathbf{v}(c)$, weights $w_A, w_Q, w_P, w_M$ & Scalar ranking score & Metrics are independently measurable; weights reflect application priorities \\
\bottomrule
\end{tabularx}
\end{table}

The algorithm enforces a deliberate ordering. Accuracy and memory are checked first because they can be evaluated cheaply (accuracy from a validation run, memory from model size and precision). Latency and energy require measurement or estimation on the target hardware, which is expensive. Filtering on cheap constraints before measuring expensive ones reduces the total evaluation cost by eliminating clearly infeasible candidates early. The Pareto filtering in line 10 removes dominated candidates, and the PEI ranking in line 11 selects the final deployment configuration.

A corollary of the monotonicity theorem is that the PEI maximizer on a finite feasible set is always Pareto optimal. The algorithm therefore produces a consistent result: it cannot select a dominated candidate.

Table~\ref{tab:analytical_models} summarizes the three analytical models, their inputs, outputs, and the assumptions under which they are valid.

%---------------------------------------------------------------------------
% Section: 6
%---------------------------------------------------------------------------
\section{Empirical Synthesis}

The preceding sections developed the taxonomy, metrics, and analytical models for deployment optimization. This section applies them to representative empirical evidence, demonstrating the cross-layer phenomena that the framework predicts. We anchor the analysis on a single well-documented benchmark study \cite{ratul2025electronics} that evaluates five inference frameworks on one hardware platform under controlled conditions, supplemented by the MLPerf Inference methodology \cite{reddi2020mlperf} as a comparability reference.

\subsection{Anchor Study and Evidence Classification}

Ratul et al.\ \cite{ratul2025electronics} evaluate PyTorch, ONNX Runtime (ORT), TensorRT (TRT), Apache TVM, and JAX on an NVIDIA Jetson AGX Orin, measuring latency, throughput, accuracy, and power for ResNet-152, MobileNetV2, and Swin-T. All measurements share the same hardware, software environment, and measurement methodology, satisfying the within-paper comparability criterion from Section~2. Latency is reported as mean single-sample inference time (model latency in the taxonomy of Section~4), throughput as offline saturation throughput, and power as platform-level draw measured via the Jetson's integrated power monitor.

The numbers in Table~\ref{tab:jetson} are measured evidence. The energy column is derived evidence computed as $E = P \cdot T$ per Eq.~\eqref{eq:energy}.

\begin{table}[t]
\centering
\caption{Inference framework comparison on NVIDIA Jetson AGX Orin. Latency is mean single-sample model latency. Throughput is offline saturation. Energy per inference is derived ($E = P \cdot T$). All values from Ratul et al.\ \cite{ratul2025electronics}.$^\dagger$}
\label{tab:jetson}
\footnotesize
\begin{tabularx}{0.505\linewidth}{>{\raggedright\arraybackslash}p{1.2cm}>{\raggedright\arraybackslash}p{1.1cm}rrrr}
\toprule
Model & Runtime & Lat.\ (ms) & Acc.\ (\%) & Tput & Power (W) \\
\midrule
R-152 & PyTorch & 9.24 & 75.3 & 932 & 21.8 \\
R-152 & ORT & 285.5 & 72.0 & 7.5 & 14.2 \\
R-152 & TRT & 2.28 & 76.6 & 652 & 28.3 \\
R-152 & TVM & 7.43 & 74.3 & 456 & 21.9 \\
R-152 & JAX & 29.1 & 72.0 & 50.3 & 15.1 \\
\midrule
MobNet & PyTorch & 4.20 & 69.7 & 1059 & 14.4 \\
MobNet & TRT & 1.14 & 70.6 & 1382 & 13.1 \\
MobNet & TVM & 9.53 & 71.2 & 171 & 13.8 \\
\midrule
Swin-T & PyTorch & 7.27 & 77.8 & 1250 & 18.8 \\
Swin-T & TRT & 3.95 & 75.7 & 318 & 27.4 \\
Swin-T & TVM & 5.28 & 77.8 & 874 & 16.0 \\
\bottomrule
\end{tabularx}
\vspace{2pt}
\par\noindent{\scriptsize $^\dagger$Offline saturation throughput depends on the maximum batch size that fits in device memory at the framework's peak allocation, not solely on per-sample latency. TensorRT's higher memory overhead per inference context can limit its maximum batch size relative to PyTorch, producing lower saturation throughput despite lower per-sample latency.}
\end{table}

\subsection{Cross-Framework Latency Spread}

The most striking observation in Table~\ref{tab:jetson} is the magnitude of the latency spread across frameworks for the same model on the same hardware. ResNet-152 exhibits a $125\times$ spread between TensorRT (2.28\,ms) and ONNX Runtime (285.5\,ms). This is not a difference between a good and a bad framework. It is a difference between a framework that applies aggressive graph fusion and hardware-specific kernel generation (TRT) and one that executes the ONNX graph with minimal optimization (ORT on this particular hardware and software version).

The spread is model dependent. MobileNetV2, a lightweight architecture designed for efficient execution, shows a much smaller $8.4\times$ spread between TRT (1.14\,ms) and TVM (9.53\,ms). Swin-T falls in between at $1.9\times$ between PyTorch (7.27\,ms) and TRT (3.95\,ms). This variation across model families is exactly the cross-layer interaction predicted in Section~3: the benefit of compiler optimization depends on how much optimization headroom the model architecture leaves. A model already designed for hardware efficiency (MobileNet) benefits less from aggressive compilation than a model designed for accuracy without hardware awareness (ResNet-152). Figure~\ref{fig:framework_latency} visualizes this spread on a log scale, making the two-order-of-magnitude gap for ResNet-152 immediately apparent.

\input{figures/fig_framework_latency}

\subsection{Accuracy Variation Across Runtimes}

A subtlety that deployment papers rarely discuss is that accuracy varies across runtimes for the same model. ResNet-152 achieves 76.6\% top-1 under TensorRT but only 72.0\% under ONNX Runtime and JAX. The 4.6 percentage point gap is not explained by the model architecture. It is a consequence of differences in precision handling, graph transformation, and numerical accumulation order across the compiler stacks.

This means that latency comparisons without accuracy normalization are incomplete. TensorRT is both the fastest and the most accurate runtime for ResNet-152 on this platform, which is unusual. More commonly, lower-precision or more aggressively optimized runtimes sacrifice accuracy. When they do, the latency gain must be evaluated against the accuracy loss on the Pareto frontier (Figure~\ref{fig:pareto}), not in isolation.

\subsection{Derived Energy per Inference}

Using the power and latency columns from Table~\ref{tab:jetson}, we compute derived energy per inference $E = P \cdot T$. Table~\ref{tab:energy} presents the results.

\begin{table}[t]
\centering
\caption{Derived energy per inference (Joules) on Jetson AGX Orin. Computed as $E = P \cdot T$ from Table~\ref{tab:jetson}. Classified as derived evidence per Section~2.}
\label{tab:energy}
\footnotesize
\begin{tabularx}{0.375\linewidth}{>{\raggedright\arraybackslash}p{1.2cm}>{\raggedright\arraybackslash}p{1.1cm}rrr}
\toprule
Model & Runtime & $P$ (W) & $T$ (ms) & $E$ (J) \\
\midrule
R-152 & PyTorch & 21.8 & 9.24 & 0.201 \\
R-152 & ORT & 14.2 & 285.5 & 4.054 \\
R-152 & TRT & 28.3 & 2.28 & 0.065 \\
R-152 & TVM & 21.9 & 7.43 & 0.163 \\
R-152 & JAX & 15.1 & 29.1 & 0.439 \\
\midrule
MobNet & TRT & 13.1 & 1.14 & 0.015 \\
Swin-T & TRT & 27.4 & 3.95 & 0.108 \\
Swin-T & TVM & 16.0 & 5.28 & 0.085 \\
\bottomrule
\end{tabularx}
\end{table}

The energy results reveal a pattern invisible in the latency numbers alone. TensorRT draws the highest power (28.3\,W for ResNet-152) but achieves the lowest energy per inference (0.065\,J) because its latency advantage more than compensates for the power increase. ONNX Runtime draws the lowest power (14.2\,W) but consumes the highest energy (4.054\,J) because its extreme latency (285.5\,ms) dominates the product.

An interesting reversal occurs for Swin-T. TVM achieves lower energy (0.085\,J) than TensorRT (0.108\,J) despite higher latency (5.28 vs 3.95\,ms), because TVM draws substantially less power (16.0 vs 27.4\,W). For energy-constrained deployments (battery powered devices, thermally throttled edge platforms), TVM would be the preferred runtime for Swin-T, contradicting the latency-only ranking. This is precisely the scenario where the multi-objective framework of Section~5 applies: the optimal choice depends on whether the constraint binds on latency or on energy. Figure~\ref{fig:energy} visualizes the reversal.

\input{figures/fig_energy}

\subsection{Implications for the Analytical Framework}

The empirical evidence from the Jetson AGX Orin study validates several predictions of the analytical framework.

The roofline model (Section~5.1) predicts that compiler optimizations primarily affect the $\eta_{\text{comp}}$ and $\eta_{\text{bw}}$ efficiency factors rather than the hardware limits $\pi_h$ and $\beta_h$. The $125\times$ latency spread across frameworks confirms this: all five runtimes execute on the same Orin hardware with the same peak FLOP/s and memory bandwidth, yet their realized efficiencies differ by two orders of magnitude.

The queuing amplification model (Section~5.2) predicts that the latency differences would produce even larger throughput differences under online serving conditions. At $\rho = 0.8$, the M/M/1 response time amplification factor is $5\times$. TensorRT's 2.28\,ms service time allows $\rho = 0.8$ at $\lambda = 351$ req/s, with mean response time 11.4\,ms. ONNX Runtime's 285.5\,ms service time allows $\rho = 0.8$ at only $\lambda = 2.8$ req/s, with mean response time 1.43\,s. The online throughput gap is not $125\times$ (the offline ratio) but effectively infinite for any reasonable latency SLA.

The Pareto analysis (Section~5.3) predicts that energy-optimal and latency-optimal configurations can differ. The Swin-T TVM vs TRT comparison confirms this: TVM is Pareto optimal for energy-constrained deployments while TRT is Pareto optimal for latency-constrained ones.

\subsection{Limitations of the Anchor Evidence}

The Ratul et al.\ study, while valuable for its controlled single-platform design, has limitations that the comparability protocol of Section~2 requires us to acknowledge.

First, the Jetson AGX Orin is an edge platform. The relative performance of frameworks may differ substantially on data center GPUs (A100, H100) where TensorRT's optimization passes have access to different hardware features (larger shared memory, NVLink, multi-instance GPU). Second, the study evaluates vision models (ResNet, MobileNet, Swin-T) rather than large language models, where the prefill/decode decomposition (Eq.~\ref{eq:autoregressive_latency}) and KV cache memory scaling (Eq.~\ref{eq:memory_decomp}) introduce qualitatively different bottlenecks \cite{kwon2023vllm}. Third, the study reports mean latency rather than tail latency percentiles, which underestimates the queuing impact at high utilization.

\input{figures/fig_datacenter_frameworks}

These limitations do not invalidate the edge evidence but restrict its scope.

\subsection{Data Center LLM Serving Evidence}

To complement the edge platform analysis, we draw on data center LLM serving benchmarks. Per the comparability policy of Section~2, cross-platform numerical comparisons are used directionally rather than literally where studies differ in model families, software versions, or measurement methodology. Where possible, we anchor claims to sources that report explicit hardware configuration, software version, and measurement protocol.

\paragraph{Hardware generation scaling.}
NVIDIA's published TensorRT-LLM benchmarks \cite{nvidia2023trtllm} report that H100 FP8 achieves up to $4.6\times$ higher maximum throughput and $4.4\times$ faster first-token latency than A100 FP16 on Llama-2-13B (TensorRT-LLM v0.5.0, TensorRT 9.1, SXM 80\,GB, TP=1, BS swept 1--64). H100 FP8 sustains over 10{,}000 output tokens per second at 100\,ms time to first token for 64 concurrent requests, while A100 FP16 saturates at approximately 2{,}200 tokens per second under the same conditions. For minimum-latency applications, H100 achieves under 10\,ms to first token at batch size~1. The hardware generation upgrade (Ampere to Hopper) accounts for roughly $2\times$ of the throughput gain through higher memory bandwidth (3.35 vs.\ 2.0\,TB/s) and improved Tensor Core throughput; the remaining $2\text{--}2.3\times$ comes from FP8 precision support and TensorRT-LLM's in-flight batching, illustrating the cross-layer nature of the improvement.

\paragraph{Serving framework comparison.}
The LLM serving framework landscape exhibits the same compiler-level variation observed on edge platforms. Table~\ref{tab:datacenter_frameworks} presents a representative H100 SXM comparison of three major serving engines evaluated on the same model under controlled conditions \cite{spheron2026benchmark}.

\begin{table}[t]
\centering
\caption{LLM serving framework comparison on H100 SXM 80\,GB. Throughput in output tokens/s; TTFT in ms. Data from Spheron \cite{spheron2026benchmark}.}
\label{tab:datacenter_frameworks}
\footnotesize
\begin{tabularx}{0.505\linewidth}{>{}p{1.3cm}rrrrrr}
\toprule
& \multicolumn{3}{c}{Throughput (tok/s)} & \multicolumn{3}{c}{TTFT p50 (ms)} \\
\cmidrule(lr){2-4} \cmidrule(lr){5-7}
Conc. & vLLM & TRT & SGLang & vLLM & TRT & SGLang \\
\midrule
1 & 120 & 130 & 125 & 45 & 38 & 42 \\
10 & 650 & 710 & 680 & 120 & 105 & 112 \\
50 & 1{,}850 & 2{,}100 & 1{,}920 & 380 & 340 & 360 \\
100 & 2{,}400 & 2{,}780 & 2{,}460 & 740 & 680 & 710 \\
\bottomrule
\end{tabularx}
\end{table}

TensorRT-LLM leads at every concurrency level, but the gap is smallest at low concurrency ($\sim$8\% at 1 request) and largest at moderate concurrency ($\sim$13\% at 50 requests). At high concurrency (100 requests), the gap compresses again to $\sim$12\%, and TTFT for all three frameworks enters the 680--740\,ms range. This convergence at high load confirms the queuing model's prediction (Section~5.2): as utilization $\rho$ approaches saturation, system-level queuing dynamics dominate over compiler-level differences. The LMSYS team reports that SGLang achieves up to $3.1\times$ higher throughput than vLLM on Llama-70B in offline scenarios \cite{zheng2024sglang}, but that advantage narrows under online serving with latency SLAs.

MLPerf Inference v5.0 \cite{mlperf2025v5} provides the most rigorous cross-submission comparison available. The benchmark introduced Llama-3.1-405B with interactive latency targets (TTFT $\le$ 2\,s, TPOT 20--50\,ms) derived from user experience studies. In the v5.1 round, open-source serving engines on H200 GPUs closed the gap with NVIDIA's proprietary implementation, reaching $\sim$90\% of NVIDIA's $\sim$35{,}000 TPS on Llama-2-70B, up from 60--80\% in the prior round. Scaling to multiple servers, vLLM v0.9.2 achieved 58{,}617 TPS on 2 servers and 87{,}334 TPS on 3 servers (1.9$\times$ and 2.8$\times$ respectively), demonstrating near-linear multi-node scaling \cite{mlperf2025v51}.

\paragraph{Kernel-level optimization at the compiler-system boundary.}
FlashAttention-3 \cite{shah2024flashattention3} exemplifies optimization at the boundary between compiler and system layers. By exploiting Hopper-specific hardware features (asynchronous Tensor Core execution, TMA-based data movement, and FP8 block quantization), FlashAttention-3 achieves 840\,TFLOP/s in BF16 (85\% utilization of H100 peak) and 1.3\,PFLOP/s in FP8, representing a $1.5\text{--}2.0\times$ speedup over FlashAttention-2. This is a concrete instance of the roofline shift described in Section~3: by tiling computation to exploit L2 and SMEM residency while overlapping compute and memory operations, FlashAttention-3 moves the attention kernel from the HBM slope toward the compute-bound regime. The gain is invisible to model-level analysis and inaccessible without hardware-specific compiler support.

\paragraph{Quantization-serving co-design.}
QServe \cite{lin2024qserve}, a W4A8KV4 quantization and serving co-design system, demonstrates that jointly optimizing quantization and the serving runtime can achieve $1.2\text{--}1.4\times$ higher throughput than TensorRT-LLM for Llama-3-8B and $2.4\text{--}3.5\times$ for Qwen1.5-72B on A100 and L40S GPUs. The key insight is that existing INT4 quantization methods suffer from 20--90\% runtime overhead during dequantization on GPU CUDA cores. QServe addresses this through progressive quantization for low-overhead W4A8 GEMM and SmoothAttention for 4-bit KV cache accuracy recovery. This result demonstrates that the quantization format and the serving engine must be co-designed: the W4A8KV4 format exploits a specific combination of weight quantization, activation precision, and KV cache compression that a general-purpose runtime cannot replicate.

\subsection{Quantization Accuracy Recovery at Scale}

Kurtic et al.\ \cite{kurtic2024bf16} conduct what is, to our knowledge, the most comprehensive quantization evaluation to date: over 500{,}000 individual evaluations across the entire Llama-3.1 model family (8B, 70B, 405B) using FP8, INT8, and INT4 formats on academic benchmarks and real-world tasks. Their findings confirm and extend the quantization analysis of Section~3.

FP8 weight-and-activation quantization (W8A8-FP) is effectively lossless across all model scales, with accuracy degradation consistently below measurement noise. INT8 (W8A8-INT) achieves surprisingly low degradation of $1\text{--}3\%$ across benchmarks. INT4 weight-only quantization (W4A16-INT) is more competitive than previously assumed, rivaling 8-bit in many scenarios. These results hold across both academic benchmarks (MMLU, HumanEval) and open-ended generation quality assessments.

The deployment implication connects directly to the memory footprint decomposition of Section~4. For Llama-3.1-70B, BF16 weights alone consume 140\,GB, requiring at least two 80\,GB A100 GPUs with tensor parallelism. INT4 quantization reduces the weight component to 35\,GB, fitting on a single GPU and eliminating the inter-device communication overhead entirely. The throughput gain from avoiding tensor parallelism communication can exceed the throughput gain from the quantization itself, a cross-layer interaction that single-layer quantization papers systematically miss.

These limitations do not invalidate the evidence but restrict its scope. Cross-platform and cross-architecture generalization with fully aligned methodology remains an open problem (Section~9).

%---------------------------------------------------------------------------
% Section: 7
%---------------------------------------------------------------------------
\section{Discussion and Limitations}

The empirical synthesis of Section~6 demonstrated that the analytical framework developed in Sections~3--5 provides explanatory power for observed deployment phenomena across both edge and data center platforms. This section examines the limitations of both the framework and the evidence base, identifies recurring patterns that cut across deployment contexts, and assesses what the current literature does well and where it falls short.

\subsection{Strengths of the Cross-Layer Perspective}

The central thesis of this paper is that deployment optimization is irreducibly cross-layer. The evidence supports this thesis from multiple angles.

On edge hardware, the $125\times$ latency spread across frameworks for a single model (Section~6.2) demonstrates that the compiler layer alone can account for two orders of magnitude in performance variation on fixed hardware. The energy reversal for Swin-T (Section~6.4, Figure~\ref{fig:energy}) demonstrates that the optimal deployment configuration changes depending on which metric is constrained, a phenomenon that single-metric benchmarks cannot capture.

On data center hardware, the QServe result (Section~6.7) demonstrates that co-designing quantization format and serving runtime yields throughput gains that neither technique achieves independently. The observation that framework performance gaps narrow under high concurrency (Section~6.7) confirms the queuing model's prediction that system-level dynamics dominate compiler-level differences near saturation.

The Kurtic et al.\ quantization study (Section~6.8) demonstrates that INT4 quantization of a 70B model can eliminate the need for tensor parallelism entirely, converting a multi-GPU deployment into a single-GPU deployment. The throughput gain from removing inter-device communication overhead is a system-level benefit triggered by a model-level optimization, precisely the cross-layer interaction that the taxonomy of Section~3 was designed to make visible.

\subsection{Limitations of the Analytical Models}

The three analytical models developed in Section~5 are useful but imperfect.

The roofline model assumes a single bottleneck memory level per kernel. In practice, modern GPU kernels access L1/SMEM, L2, and HBM simultaneously, and the effective bandwidth depends on the data reuse pattern, which is kernel-specific. The model also assumes no overlap between compute and memory operations, which underestimates the performance of well-pipelined kernels. These limitations mean that the roofline bound is conservative: it correctly identifies the bottleneck resource but may overestimate the gap between the bound and achieved performance.

The queuing model (M/G/1) assumes Poisson arrivals and a single logical server. Real inference serving systems use continuous batching, which effectively creates a multi-server system with correlated service times. The M/G/1 model captures the qualitative shape of the response time curve (the hockey-stick amplification near saturation) but underestimates tail latency at moderate utilization and overestimates it at very high utilization where admission control intervenes. A more accurate model would require workload-specific simulation, which sacrifices the analytical tractability that makes the M/G/1 framework useful for first-order reasoning.

The PEI ranking functional assumes that all metrics in the vector are independently measurable and that the weight vector accurately reflects application priorities. In practice, metrics interact: reducing latency may increase power (as observed for TensorRT on the Orin), and reducing memory footprint via quantization may change accuracy in ways that depend on the input distribution. The PEI functional treats these interactions as external to the ranking, which is correct for a given measurement but may mislead if the measurements themselves shift when the deployment configuration changes.

\subsection{Limitations of the Evidence Base}

The evidence synthesized in this paper has structural limitations that reflect the state of the field rather than choices specific to this study.

\paragraph{Vendor concentration.} The controlled empirical evidence (Section~6.1--6.5) comes from a single NVIDIA platform. The data center evidence (Section~6.7--6.8) also focuses on NVIDIA GPUs (A100, H100, L40S). AMD, Intel, Google TPU, and custom accelerator deployments are underrepresented in the literature that meets our inclusion criteria. This is not because those platforms are unimportant but because the published studies that report latency, throughput, accuracy, power, and software version metadata simultaneously are overwhelmingly NVIDIA-centric.

\paragraph{Vision model bias.} The controlled Jetson study evaluates vision models (ResNet, MobileNet, Swin-T). While the data center subsections address LLM serving, the paper lacks a controlled within-study comparison of LLM serving frameworks with the same rigor as the Ratul et al.\ edge study. The prefill/decode asymmetry (Eq.~\ref{eq:autoregressive_latency}), KV cache scaling (Eq.~\ref{eq:memory_decomp}), and continuous batching dynamics of LLM serving create deployment bottlenecks that differ qualitatively from vision model inference. A controlled multi-framework LLM serving benchmark with power measurement would significantly strengthen the empirical foundation.

\paragraph{Missing variance reporting.} Nearly all sources cited in this paper report mean latency or median throughput without confidence intervals, standard deviations, or percentile distributions. This makes it impossible to assess whether observed differences between frameworks are statistically significant or within measurement noise. The comparability policy of Section~2 mitigates this by restricting literal comparisons to within-study results, but even within-study comparisons are weakened by the absence of variance data.

\paragraph{Temporal fragility.} Deployment benchmarks are perishable. A TensorRT update, a new CUDA toolkit, or a vLLM release can shift framework rankings by $10\text{--}30\%$ (Table~\ref{tab:comparability}). The specific numbers in Section~6 reflect a snapshot in time. The analytical framework and the qualitative patterns (cross-layer entanglement, energy reversals, queuing amplification) are durable; the specific numerical rankings are not.

\subsection{What the Literature Does Well}

Despite the limitations above, the deployment literature has made substantial progress in three areas.

First, the MLPerf Inference benchmark \cite{reddi2020mlperf} has established a credible methodology for standardized comparison, with defined scenarios, quality targets, and reporting requirements. Studies that follow MLPerf methodology produce results that are meaningfully comparable within the constraints of their hardware and software configurations.

Second, the systems community has converged on continuous batching \cite{yu2022orca} and paged KV cache management \cite{kwon2023vllm} as foundational techniques for LLM serving. These are no longer experimental; they are production infrastructure. The analytical framework of this paper can take them as given when modeling online serving throughput.

Third, the quantization community has moved beyond accuracy-only evaluation. The Kurtic et al.\ study \cite{kurtic2024bf16} evaluates quantization across 500{,}000+ configurations, and the QServe system \cite{lin2024qserve} co-designs quantization with the serving runtime. This cross-layer thinking, which was rare five years ago, is becoming the norm in systems-aware quantization research.

\subsection{What the Literature Still Lacks}

Figure~\ref{fig:evidence_coverage} visualizes these gaps as a radar chart across the seven deployment metrics. Three gaps remain conspicuous.

\input{figures/fig_evidence_coverage}

First, there is no published study that evaluates the same model on the same serving framework across three or more hardware platforms (e.g., Jetson Orin, A100, H100, AMD MI300X) with aligned methodology and power measurement. Such a study would enable the cross-platform generalization that this paper can only discuss qualitatively.

Second, energy per inference is almost never reported as a first-class metric. Power measurements, when they appear, are typically accelerator-only and instantaneous rather than platform-level and sustained. The energy decomposition of Eq.~\ref{eq:power_decomp} cannot be applied to most published results because the component terms are not reported.

Third, compilation cost $C(c)$ remains the least reported metric in the deployment vector. TVM autotuning campaigns can require hours of GPU time; TensorRT engine builds for large models can take tens of minutes. These costs directly affect deployment velocity but are systematically omitted from benchmark papers, making it impossible to include compilation cost in the PEI ranking without custom measurement.

%---------------------------------------------------------------------------
% Section: 8
%---------------------------------------------------------------------------
\section{Design Guidance for Practitioners}

The analytical framework and empirical evidence of the preceding sections converge on a set of actionable principles for deployment engineers. This section distills those principles into concrete guidance, organized around the three decision points that every deployment team faces: what to measure, how to eliminate infeasible candidates, and how to select among survivors.

\subsection{What to Measure Before Deciding}

The most common deployment failure is optimizing the wrong metric. The metric vector of Section~4 defines seven quantities, but not all are equally expensive to obtain. We recommend a staged measurement protocol.

\input{figures/fig_decision_flow}

\paragraph{Stage 1: Free metrics.} Model size ($N_{\text{params}} \cdot w_p$) and peak memory footprint $M(c)$ can be computed from the model architecture and precision format without running any hardware. These determine whether a candidate fits on the target device at all. Candidates that exceed the available HBM should be eliminated immediately.

\paragraph{Stage 2: Cheap metrics.} Accuracy $A(c)$ requires a validation run but not a serving deployment. A single forward pass over the evaluation set on any compatible hardware is sufficient. Candidates that fall below $A_{\min}$ should be eliminated before investing in latency measurement.

\paragraph{Stage 3: Expensive metrics.} Latency $T(c)$, throughput $Q(c)$, power $P(c)$, and energy $E(c)$ require measurement on the target hardware under realistic conditions. These should be measured only for candidates that survive Stages 1 and 2. Algorithm~\ref{alg:selection} formalizes this ordering.

This staged approach can reduce the total number of hardware measurements by an order of magnitude. If 50 candidate configurations are under consideration and 35 are eliminated by memory or accuracy constraints, only 15 require expensive on-device profiling. Figure~\ref{fig:decision_flow} visualizes this workflow as a decision flowchart.

\subsection{Worked Example: Applying the Selection Algorithm}

To make the staged protocol concrete, we apply Algorithm~\ref{alg:selection} to the Jetson AGX Orin data from Table~\ref{tab:jetson} with the following practitioner constraints: $A_{\min} = 73\%$, $T_{\max} = 10$\,ms, $E_{\max} = 0.25$\,J. We set equal PEI weights ($w_A = w_Q = w_P = w_M = 1$) for simplicity. Table~\ref{tab:worked} traces the elimination.

\begin{table}[t]
\centering
\caption{Worked selection example on Jetson AGX Orin data. Constraints: $A_{\min}=73\%$, $T_{\max}=10$\,ms, $E_{\max}=0.25$\,J. Candidates failing any constraint are struck through.}
\label{tab:worked}
\footnotesize
\begin{tabularx}{0.43\linewidth}{>{\raggedright\arraybackslash}p{1.0cm}>{\raggedright\arraybackslash}p{0.9cm}rrrl}
\toprule
Model & Fwk & Acc & $T$ (ms) & $E$ (J) & Result \\
\midrule
R-152 & TRT & 76.6 & 2.28 & 0.065 & $\checkmark$ Pareto \\
R-152 & TVM & 74.3 & 7.43 & 0.163 & $\checkmark$ Feasible \\
R-152 & PyT & 75.3 & 9.24 & 0.201 & $\checkmark$ Feasible \\
R-152 & ORT & 72.0 & 285.5 & 4.054 & \textcolor{papersecondary}{\texttimes\ Acc.} \\
R-152 & JAX & 72.0 & 29.1 & 0.439 & \textcolor{papersecondary}{\texttimes\ Acc.} \\
\midrule
Mob & TRT & 70.6 & 1.14 & 0.015 & \textcolor{papersecondary}{\texttimes\ Acc.} \\
Mob & PyT & 69.7 & 4.20 & 0.060 & \textcolor{papersecondary}{\texttimes\ Acc.} \\
Mob & TVM & 71.2 & 9.53 & 0.132 & \textcolor{papersecondary}{\texttimes\ Acc.} \\
\midrule
Swin & PyT & 77.8 & 7.27 & 0.137 & $\checkmark$ Feasible \\
Swin & TVM & 77.8 & 5.28 & 0.085 & $\checkmark$ Pareto \\
Swin & TRT & 75.7 & 3.95 & 0.108 & $\checkmark$ Feasible \\
\bottomrule
\end{tabularx}
\end{table}

Of 11 candidates, 5 are eliminated by accuracy ($A < 73\%$), and all MobileNetV2 configurations fail this gate despite having the lowest latency and energy. The 6 survivors all satisfy $T_{\max}$ and $E_{\max}$. Pareto filtering identifies R-152/TRT (best latency and energy among ResNets) and Swin-T/TVM (highest accuracy with lowest energy among Swin-Ts) as non-dominated. The PEI maximizer selects the final candidate based on the weight vector: with equal weights, Swin-T/TVM wins on combined accuracy-energy; with latency-heavy weights, R-152/TRT wins.

The key insight is that all MobileNetV2 configurations, despite being the fastest and most energy-efficient, are eliminated in Stage~2 because they fail the accuracy constraint. Speed without accuracy is useless.

\subsection{Decision Tree for Common Scenarios}

The following decision logic covers the most frequently encountered deployment scenarios.

\paragraph{Latency-critical, single-device (e.g., real-time inference endpoint).}
Filter by $M(c) \le M_{\text{HBM}}$ and $A(c) \ge A_{\min}$. Among survivors, select the configuration with lowest $T(c)$. If multiple configurations achieve similar latency ($<5\%$ difference), break ties on energy $E(c)$. TensorRT or framework-specific AOT compilation is typically optimal in this regime.

\paragraph{Throughput-critical, multi-tenant (e.g., batch processing or high-QPS serving).}
Filter by $M(c) \le M_{\text{HBM}}$ and $A(c) \ge A_{\min}$. Among survivors, select the configuration with highest $Q(c)$ under the target SLA using Eq.~\ref{eq:max_throughput}. Continuous batching \cite{yu2022orca} and PagedAttention \cite{kwon2023vllm} are prerequisites. The queuing amplification of Eq.~\ref{eq:sensitivity} means that even small per-request latency reductions compound into large throughput gains near saturation.

\paragraph{Memory-constrained (e.g., edge device or single consumer GPU).}
The binding constraint is $M(c) \le M_{\text{device}}$. Quantization is the primary lever: INT4 weight-only quantization reduces the weight component by $4\times$ relative to FP16 (Section~6.8). If the quantized model fits on a single device while the unquantized model requires tensor parallelism, the throughput gain from eliminating inter-device communication can exceed the gain from the quantization itself.

\paragraph{Energy-constrained (e.g., battery-powered edge, thermally throttled SoC).}
Optimize $E(c) = P(c) \cdot T(c)$ rather than $T(c)$ alone. As the Swin-T example demonstrates (Section~6.4, Figure~\ref{fig:energy}), the energy-optimal framework may differ from the latency-optimal one. Lower-power runtimes that sacrifice some latency can achieve lower total energy per inference.

\subsection{When to Re-evaluate}

Deployment configurations are not permanent. Three events should trigger re-evaluation.

\paragraph{Model update.} Retraining, fine-tuning, or architecture changes alter the computational graph, invalidating compiler-specific optimizations (autotuned tile sizes, fused kernel configurations). Re-profile after any model change.

\paragraph{Hardware refresh.} Migrating from A100 to H100 changes peak FLOP/s, memory bandwidth, and available precision formats (FP8 on Hopper). The roofline parameters $\pi_h$ and $\beta_h$ change, potentially shifting kernels between memory-bound and compute-bound regimes. Re-profile on the new hardware.

\paragraph{Workload shift.} Changes in request arrival rate $\lambda$, sequence length distribution, or batch size distribution alter the queuing operating point $\rho$ and the effective service time distribution $C_s$. A configuration that was optimal at $\rho = 0.5$ may be suboptimal at $\rho = 0.8$ where the sensitivity factor $1/(1-\rho)^2$ is $6.25\times$ larger.

\subsection{Anti-Patterns to Avoid}

We conclude with three deployment anti-patterns that the framework exposes.

\paragraph{Benchmark on batch size 1, deploy on batch size 32.} Latency at batch size 1 is dominated by kernel launch overhead and memory latency. Latency at batch size 32 is dominated by compute throughput and memory bandwidth. The ranking of frameworks can and does reverse between these two regimes. Always profile at the deployment batch size.

\paragraph{Optimize latency, ignore energy.} As Section~6.4 demonstrates, the fastest framework may not be the most energy efficient. For deployments where total energy cost matters (cloud billing, battery life, thermal envelope), energy per inference must be measured and optimized explicitly.

\paragraph{Quantize without re-profiling the serving stack.} INT4 quantization changes operational intensity (Eq.~\ref{eq:oi_precision}), memory footprint (Eq.~\ref{eq:memory_decomp}), and potentially the number of GPUs required. These changes propagate through the compiler and system layers in ways that cannot be predicted from the model layer alone. The QServe result \cite{lin2024qserve} demonstrates that co-designing quantization with the serving runtime yields $2\text{--}3.5\times$ higher throughput than applying quantization to a general-purpose runtime.

%---------------------------------------------------------------------------
% Section: 9
%---------------------------------------------------------------------------
\section{Open Research Problems}

The gaps identified in Section~7 and the practitioner challenges cataloged in Section~8 point to five concrete research directions. Each represents a problem where the analytical framework of this paper defines what a solution would look like but where the necessary tools, data, or methodology do not yet exist. Figure~\ref{fig:research_roadmap} maps these problems to the deployment layers they span.

\input{figures/fig_research_roadmap}

\subsection{Joint Compiler-Serving Co-optimization}

Current deployment pipelines treat compilation and serving as sequential, independent stages: the compiler produces an optimized binary, and the serving system executes it. The QServe result (Section~6.7) demonstrates that co-designing these stages yields $2\text{--}3.5\times$ throughput gains. Yet no general-purpose framework exists for jointly optimizing compiler decisions (precision format, fusion strategy, kernel selection) and serving decisions (batch size, admission policy, KV cache allocation) in a single optimization loop.

The technical challenge is that the search space is combinatorially large. A compiler may offer $10^3$ kernel configurations per operator, and a serving system may offer $10^2$ batching and scheduling configurations. Joint optimization over $10^5$ combinations per operator is intractable by exhaustive search. Promising directions include Bayesian optimization over the joint space, reinforcement learning with compiler and serving actions as a unified action space, and hierarchical decomposition where the compiler optimizes locally and the serving system optimizes globally.

\subsection{Cross-Hardware Latency Prediction}

The roofline model (Section~5.1) provides hardware-parameterized bounds, but predicting actual latency on a new hardware target from measurements on an existing target remains unsolved. A model profiled on an A100 cannot reliably predict its latency on an H100 because differences in shared memory size, warp scheduler behavior, tensor core microarchitecture, and L2 cache policy produce nonlinear performance shifts that the roofline's linear bounds cannot capture \cite{lee2021help}.

A practical solution would enable deployment teams to evaluate candidate hardware without physical access, reducing the cost of hardware selection from weeks of profiling to hours of prediction. Transfer learning approaches that fine-tune latency predictors across hardware generations show promise but require large profiling datasets that are expensive to collect.

\subsection{Standardized Energy and Variance Reporting}

As Figure~\ref{fig:evidence_coverage} illustrates, energy per inference and compilation cost are the two most underreported metrics in the deployment literature. The energy gap is particularly consequential for edge and sustainability-focused deployments.

We propose a minimal reporting standard for deployment benchmarks:
\begin{enumerate}[leftmargin=1.5em]
  \item Platform-level sustained power draw (not instantaneous, not accelerator-only) measured over at least 60 seconds of continuous inference.
  \item Latency reported as mean, p50, p95, and p99 with at least 1{,}000 inference iterations after warmup.
  \item Energy per inference computed as $E = \bar{P} \cdot \bar{T}$ with both components reported.
  \item Compilation time reported as wall-clock time for the full optimization pipeline.
  \item Hardware and software metadata: GPU model, driver version, CUDA toolkit version, framework version, precision format, batch size.
\end{enumerate}

Adopting this standard would enable the PEI ranking of Section~5.3 to be applied to published results without custom measurement, transforming the analytical framework from a conceptual tool to a practical one.

\subsection{Reproducible Benchmark Artifacts}

The temporal fragility identified in Section~7.3 means that benchmark results are perishable. A result published with TensorRT 8.6 may not reproduce with TensorRT 9.0. The MLPerf Inference benchmark \cite{reddi2020mlperf} addresses this partially through versioned submissions, but most published studies do not provide containerized or otherwise reproducible execution environments.

A research infrastructure contribution would be a standardized container specification for inference benchmarks that pins all software dependencies (driver, toolkit, framework, model weights, calibration data) and produces deterministic results across runs. This would convert the comparability policy of Section~2 from a post-hoc analysis rule to a pre-registration requirement.

\subsection{LLM-Specific Deployment Models}

The queuing and roofline models developed in Section~5 are general-purpose. LLM serving introduces domain-specific phenomena that these models do not capture: the prefill/decode phase asymmetry (Eq.~\ref{eq:autoregressive_latency}), KV cache growth with sequence length (Eq.~\ref{eq:memory_decomp}), speculative decoding, and prefix caching. Each of these creates optimization opportunities and failure modes that a general-purpose framework misses.

Extending the analytical framework to LLM-specific models would require:
\begin{enumerate}[leftmargin=1.5em]
  \item A two-phase roofline that separately bounds prefill (compute-bound, high OI) and decode (memory-bound, low OI) latency.
  \item A queuing model that accounts for variable-length service times with heavy-tailed distributions ($C_s \gg 1$), correlated arrivals from multi-turn conversations, and the memory pressure of KV cache accumulation.
  \item A memory-aware admission controller that dynamically adjusts the maximum batch size based on available KV cache headroom rather than a fixed concurrency limit.
\end{enumerate}

These extensions would connect the general framework of this paper to the specific bottlenecks that dominate production LLM serving.

\section{Conclusion}

The performance of a deployed AI model is not determined by the model alone, nor by the compiler alone, nor by the serving system alone. It is determined by their interaction. This paper has developed the analytical vocabulary, formal models, and empirical evidence to make that claim precise.

We introduced a three-layer deployment taxonomy (model, compiler, system) and demonstrated that cross-layer interactions produce phenomena that no single-layer analysis can predict: a $125\times$ latency spread across frameworks for the same model on the same hardware (Section~6.2), an energy ranking reversal where the fastest framework is not the most efficient (Section~6.4), and a quantization-induced elimination of tensor parallelism that converts a multi-GPU deployment into a single-GPU one (Section~6.8). We formalized inference optimization using three analytical models. The roofline model (Eq.~\ref{eq:roofline_time}) identifies whether a kernel is compute-bound or memory-bound, preventing wasted optimization effort on the wrong bottleneck. The M/G/1 queuing model (Eq.~\ref{eq:mg1}) quantifies the nonlinear amplification of service time into response time, showing that a 1\,ms latency reduction at $\rho = 0.9$ yields a 100\,ms response time improvement. The Performance Efficiency Index (Eq.~\ref{eq:pei}) provides a monotonic, scale-invariant ranking functional for multi-objective deployment selection, with a constraint-aware algorithm (Algorithm~\ref{alg:selection}) that filters candidates by cost-ordered stages.

We enforced a comparability protocol that classifies every claim as measured, derived, or analytical evidence and restricts literal numerical comparison to within-study results. This protocol exposed why cross-paper benchmark comparisons are hazardous: latency definitions, batch semantics, power scope, and software versions differ in ways that can reverse framework rankings (Table~\ref{tab:comparability}). In addition, we synthesized empirical evidence spanning edge platforms \cite{ratul2025electronics} and data center GPUs, including H100 vs.\ A100 throughput scaling, LLM serving framework comparisons \cite{kwon2023vllm}, quantization-serving co-design \cite{lin2024qserve}, and the largest published quantization evaluation to date \cite{kurtic2024bf16}. The evidence consistently confirms the cross-layer thesis: deployment outcomes are governed by interactions between model compression, compiler transformations, and serving policies that cannot be predicted from any single layer in isolation.

Five open problems remain (Section~9): joint compiler-serving co-optimization, cross-hardware latency prediction, standardized energy and variance reporting, reproducible benchmark artifacts, and LLM-specific deployment models. Each represents a gap where the analytical framework defines what a solution would look like but where the necessary tools do not yet exist. The central message is simple. Deployment optimization is not a framework ranking problem. It is a constrained, multi-objective, cross-layer systems problem. The paper that treats it as such with the right taxonomy, metrics, analytical models, and evidence protocol provides the foundation for every deployment decision that follows.

\section{Acknowledgment}
The authors would like to thank Michael Robillard and Krish Iyer (Office of the CTO, Dell Technologies) for their insightful feedback and stimulating discussions that significantly contributed to the development of this technical review.

\bibliographystyle{IEEEtran}
\bibliography{references}

\end{document}

%% file: figures/fig_taxonomy.tex
% Figure 1: Deployment pipeline with cross-layer optimization
\begin{figure*}[t]
\centering
\begin{tikzpicture}[
  >=Latex,
  font=\sffamily\small,
  stage/.style={
    font=\sffamily\scriptsize, rectangle, rounded corners=2pt, thick,
    minimum height=1.0cm, minimum width=1.8cm, align=center,
    draw=#1!70!black, fill=#1!12, text=paperblack
  },
  annot/.style={
    font=\sffamily\scriptsize, text=papermuted, align=center
  },
  optlabel/.style={
    font=\sffamily\bfseries\footnotesize, text=#1!70!black
  },
  arr/.style={-{Latex[length=2.0mm]}, line width=0.9pt, draw=paperaxis}
]

% ---- Pipeline stages (left to right, tighter spacing) ----
\node[stage=paperprimary] (train) {Trained\\Model};
\node[stage=papersecondary, right=0.7cm of train] (compress) {Compress\\\& Quantize};
\node[stage=papertertiary, right=0.7cm of compress] (graphir) {Graph IR\\Optimize};
\node[stage=papertertiary, right=0.7cm of graphir] (codegen) {Kernel\\Code Gen};
\node[stage=paperaccent, right=0.7cm of codegen] (runtime) {Runtime\\Execution};
\node[stage=paperaccent, right=0.7cm of runtime] (serving) {Serving\\System};

% ---- Forward arrows ----
\draw[arr] (train) -- (compress);
\draw[arr] (compress) -- (graphir);
\draw[arr] (graphir) -- (codegen);
\draw[arr] (codegen) -- (runtime);
\draw[arr] (runtime) -- (serving);

% ---- Optimization layer brackets (below) ----
\draw[decorate, decoration={brace, amplitude=4pt, mirror}, thick, draw=papersecondary!70!black]
  ([yshift=-2pt]train.south west) -- ([yshift=-2pt]compress.south east)
  node[midway, below=6pt, optlabel=papersecondary] {Model Level};

\draw[decorate, decoration={brace, amplitude=4pt, mirror}, thick, draw=papertertiary!70!black]
  ([yshift=-2pt]graphir.south west) -- ([yshift=-2pt]codegen.south east)
  node[midway, below=6pt, optlabel=papertertiary] {Compiler Level};

\draw[decorate, decoration={brace, amplitude=4pt, mirror}, thick, draw=paperaccent!70!black]
  ([yshift=-2pt]runtime.south west) -- ([yshift=-2pt]serving.south east)
  node[midway, below=6pt, optlabel=paperaccent] {System Level};

% ---- Technique annotations (above) ----
\node[annot, above=3pt of compress] {Pruning, Distillation,\\Mixed Precision};
\node[annot, above=3pt of graphir] {Fusion, Layout\\Transform, Folding};
\node[annot, above=3pt of codegen] {Autotuning,\\Tensor Core Map};
\node[annot, above=3pt of runtime] {Memory Tiering,\\Kernel Dispatch};
\node[annot, above=3pt of serving] {Batching, Admission\\Control, Scaling};

% ---- Feedback arrow (clean path below the braces) ----
\draw[arr, densely dashed, draw=papermuted]
  (serving.south) -- ++(0,-1.2)
  -| (compress.south)
  node[pos=0.5, below=2pt, annot] {Deployment feedback: accuracy / latency / energy telemetry};

\end{tikzpicture}
\caption{End to end deployment pipeline. A trained model passes through model level compression, compiler level graph and kernel optimization, and system level serving. The dashed feedback path represents iterative retuning driven by deployment telemetry.}
\label{fig:taxonomy}
\end{figure*}
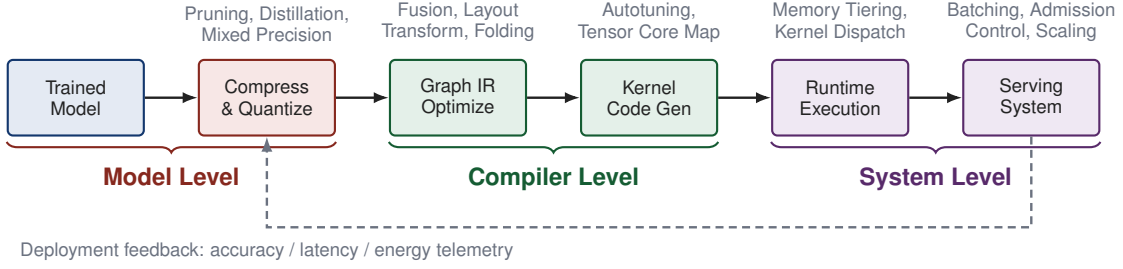

%% file: figures/fig_roofline.tex
% Figure 2: Roofline model — NVIDIA A100 SXM4
% Multiple memory bandwidth tiers shown as parallel ascending slopes.
% A single FP16 TC compute ceiling. FP32 shown as thin reference.
% Ridge points mark the transition from memory-bound to compute-bound.
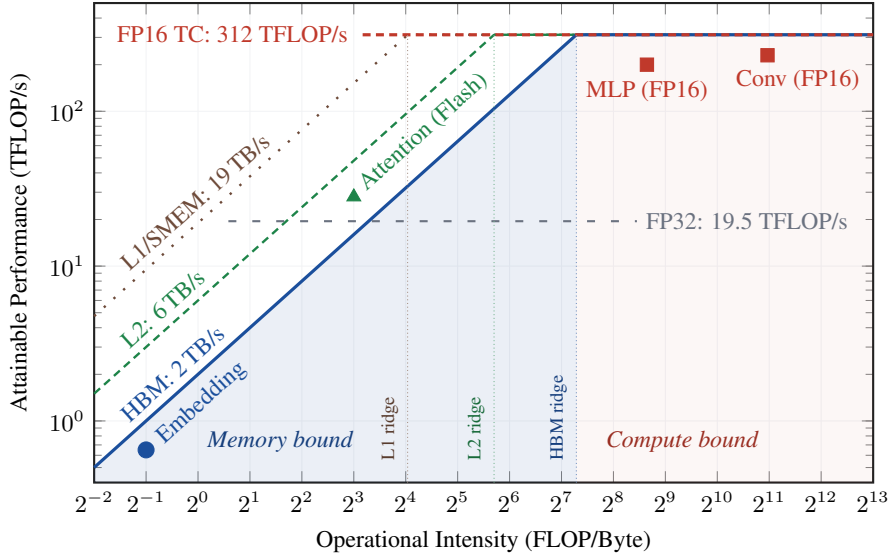
\begin{figure}[t]
\centering
\begin{tikzpicture}
\begin{axis}[
  width=0.72\linewidth,
  height=0.48\linewidth,
  xlabel={Operational Intensity (FLOP/Byte)},
  ylabel={Attainable Performance (TFLOP/s)},
  xmode=log, ymode=log,
  log basis x=2,
  xmin=0.25, xmax=8192,
  ymin=0.4, ymax=500,
  grid=major,
  grid style={draw=papergrid!50, line width=0.3pt},
  axis line style={draw=paperaxis, thick},
  tick label style={font=\small},
  label style={font=\small},
  clip=false
]

% ==== Memory-bound region: triangle below HBM slope, left of HBM ridge ====
\fill[paperprimary, opacity=0.09]
  (axis cs:0.25,0.4) -- (axis cs:0.25,0.5)
  -- (axis cs:156,312) -- (axis cs:156,0.4) -- cycle;

% ==== Compute-bound region: rectangle right of HBM ridge, below FP16 ceiling ====
\fill[papersecondary, opacity=0.04]
  (axis cs:156,0.4) -- (axis cs:156,312)
  -- (axis cs:8192,312) -- (axis cs:8192,0.4) -- cycle;

% ==== Memory bandwidth slopes (parallel, different BW) ====
% L1/Shared memory: 19 TB/s
\addplot[domain=0.25:8192, samples=500, color=paperquaternary, line width=0.9pt, loosely dotted]
  {min(19*x, 312)};

% L2 cache: 6 TB/s
\addplot[domain=0.25:8192, samples=500, color=papertertiary, line width=0.9pt, densely dashed]
  {min(6*x, 312)};

% HBM: 2 TB/s
\addplot[domain=0.25:8192, samples=500, color=paperprimary, line width=1.2pt, solid]
  {min(2*x, 312)};

% ==== Compute ceilings ====
% FP32: 19.5 TFLOP/s (thin reference)
\addplot[domain=1.5:350, samples=2, color=papermuted, line width=0.8pt, loosely dashed]
  {19.5};
% FP16 TC: 312 TFLOP/s (main ceiling)
\addplot[domain=9:8192, samples=2, color=papersecondary, line width=1.2pt, densely dashed]
  {312};

% ==== Bandwidth slope labels ====
% With xmax=8192 (15 octaves) and same y-range (3.1 decades),
% the visual slope angle is atan(dy/dx) = atan(1.0) ≈ 45 degrees
\node[font=\small, text=paperquaternary, rotate=42, anchor=south west]
  at (axis cs:0.4, 7.6) {L1/SMEM: 19\,TB/s};
\node[font=\small, text=papertertiary, rotate=42, anchor=south west]
  at (axis cs:0.4, 2.4) {L2: 6\,TB/s};
\node[font=\small, text=paperprimary, rotate=42, anchor=south west]
  at (axis cs:0.4, 0.8) {HBM: 2\,TB/s};

% ==== Compute ceiling labels (right end of flat portion) ====
\node[font=\small, text=papermuted, anchor=south west]
  at (axis cs:350, 15) {FP32: 19.5 TFLOP/s};
\node[font=\small, text=papersecondary, anchor=south west]
  at (axis cs:0.3, 242) {FP16 TC: 312 TFLOP/s};

% ==== Ridge point vertical lines and labels ====
\draw[thin, densely dotted, draw=paperquaternary!60]
  (axis cs:16.4, 0.4) -- (axis cs:16.4, 312);
\draw[thin, densely dotted, draw=papertertiary!60]
  (axis cs:52, 0.4) -- (axis cs:52, 312);
\draw[thin, densely dotted, draw=paperprimary!60]
  (axis cs:156, 0.4) -- (axis cs:156, 312);

% Ridge labels inside plot
\node[font=\scriptsize, text=paperquaternary!80!black, rotate=90, anchor=south]
  at (axis cs:16.4, 0.9) {L1 ridge};
\node[font=\scriptsize, text=papertertiary!80!black, rotate=90, anchor=south]
  at (axis cs:52, 0.9) {L2 ridge};
\node[font=\scriptsize, text=paperprimary!80!black, rotate=90, anchor=south]
  at (axis cs:156, 1.05) {HBM ridge};

% ==== Region labels ====
\node[font=\small\itshape, text=paperprimary!80!black, anchor=north west]
  at (axis cs:1, 1) {Memory bound};
\node[font=\small\itshape, text=papersecondary!80!black, anchor=north east]
  at (axis cs:2000, 1) {Compute bound};

% ==== Kernel workload points ====
% Embedding: deep HBM-bound, below HBM slope
\addplot[only marks, mark=*, mark size=3pt, color=paperprimary, forget plot]
  coordinates {(0.5, 0.65)};
\node[font=\small, text=paperprimary, anchor=west, rotate=42, xshift=4pt]
  at (axis cs:0.5, 0.62) {Embedding};

% Attention with FlashAttention: exploits L2 SRAM, above HBM roof at this OI
\addplot[only marks, mark=triangle*, mark size=3pt, color=papertertiary, forget plot]
  coordinates {(8, 28)};
\node[font=\small, text=papertertiary, anchor=south west, rotate=42, xshift=3pt]
  at (axis cs:9, 22) {Attention (Flash)};

% MLP FP16: compute-bound (OI=400, well right of HBM ridge at 156)
% Large batch transformer FFN: large OI, achieved ~65% of FP16 ceiling
\addplot[only marks, mark=square*, mark size=2.5pt, color=papersecondary, forget plot]
  coordinates {(400, 200)};
\node[font=\small, text=papersecondary, anchor=south, yshift=3pt]
  at (axis cs:400, 90) {MLP (FP16)};

% Conv FP16: deep compute-bound
% High OI large conv layer: achieved ~70% of FP16 ceiling
\addplot[only marks, mark=square*, mark size=2.5pt, color=papersecondary, forget plot]
  coordinates {(2000, 230)};
\node[font=\small, text=papersecondary, anchor=south, yshift=3pt]
  at (axis cs:3000, 105) {Conv (FP16)};

\end{axis}
\end{tikzpicture}
\caption{Roofline model for NVIDIA A100 SXM4 illustrating the memory bandwidth hierarchy (L1/SMEM at 19\,TB/s, L2 at 6\,TB/s, HBM at 2\,TB/s) and the FP16 Tensor Core compute ceiling (312 TFLOP/s). Vertical dotted lines mark ridge points for each memory tier. Kernels operating below the HBM slope are HBM-bound; those between the HBM and L2 slopes benefit from L2 data reuse, as illustrated by FlashAttention. Kernels to the right of the FP16 ridge are compute-bound \cite{williams2009roofline}.}
\label{fig:roofline}
\end{figure}

%% file: figures/fig_pareto.tex
% Figure 3: Pareto frontier of deployment configurations
% Each point = a specific (model, precision, compiler, batch) config.
% Non-dominated points form the frontier. Dominated points fall below it.
% Feasibility constraints define the viable deployment region.
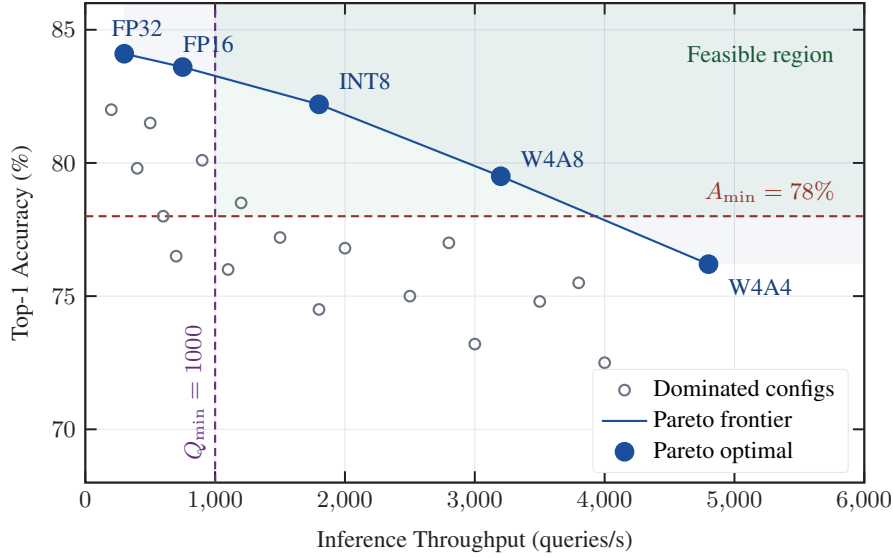
\begin{figure}[t]
\centering
\begin{tikzpicture}
\begin{axis}[
  paper line plot,
  width=0.72\linewidth,
  height=0.48\linewidth,
  xlabel={Inference Throughput (queries/s)},
  ylabel={Top-1 Accuracy (\%)},
  xmin=0, xmax=6000,
  ymin=68, ymax=86,
  grid=major,
  grid style={draw=papergrid, line width=0.4pt},
  legend style={at={(0.98,0.02)}, anchor=south east, font=\small,
    draw=papergrid, fill=paperwhite, rounded corners=1.5pt},
  legend cell align=left,
  tick label style={font=\small},
  label style={font=\small},
  clip=false
]

% ---- Shaded infeasible region above frontier ----
\fill[paperprimary, opacity=0.05]
  (axis cs:300,84.1) -- (axis cs:750,83.6) -- (axis cs:1800,82.2)
  -- (axis cs:3200,79.5) -- (axis cs:4800,76.2)
  -- (axis cs:6000,76.2) -- (axis cs:6000,86) -- (axis cs:300,86) -- cycle;

% ---- Feasibility constraint: minimum accuracy ----
\draw[thick, densely dashed, draw=papersecondary!80!black]
  (axis cs:0,78) -- (axis cs:6000,78);
\node[font=\small, text=papersecondary!80!black, anchor=south west]
  at (axis cs:4700,78.2) {$A_{\min}=78\%$};

% ---- Feasibility constraint: minimum throughput ----
\draw[thick, densely dashed, draw=paperaccent!80!black]
  (axis cs:1000,68) -- (axis cs:1000,86);
\node[font=\small, text=paperaccent!80!black, anchor=south west, rotate=90]
  at (axis cs:1000,68.5) {$Q_{\min}=1000$};

% ---- Shaded feasible region ----
\fill[papertertiary, opacity=0.06]
  (axis cs:1000,78) rectangle (axis cs:6000,86);
\node[font=\small, text=papertertiary!70!black] at (axis cs:5200,84) {Feasible region};

% ---- Dominated candidates: various (model, compiler, precision, batch) configs ----
% These represent real evaluated configs that are strictly inferior to
% at least one Pareto-optimal config on both accuracy and throughput.
\addplot[only marks, mark=o, mark size=2pt, thick, color=papermuted] coordinates {
  (200, 82.0) (500, 81.5) (400, 79.8) (900, 80.1) (1200, 78.5)
  (1500, 77.2) (2000, 76.8) (2500, 75.0) (1800, 74.5) (3000, 73.2)
  (3500, 74.8) (4000, 72.5) (4500, 71.0) (2800, 77.0) (1100, 76.0)
  (600, 78.0) (3800, 75.5) (5200, 70.5) (700, 76.5)
};
\addlegendentry{Dominated configs}

% ---- Pareto frontier line ----
\addplot[color=paperprimary, thick, mark=none] coordinates {
  (300,84.1) (750,83.6) (1800,82.2) (3200,79.5) (4800,76.2)
};
\addlegendentry{Pareto frontier}

% ---- Pareto-optimal candidates ----
\addplot[only marks, mark=*, mark size=3pt, color=paperprimary] coordinates {
  (300, 84.1) (750, 83.6) (1800, 82.2) (3200, 79.5) (4800, 76.2)
};
\addlegendentry{Pareto optimal}

% ---- Labels for Pareto-optimal points (offset to avoid overlap) ----
\node[font=\small, anchor=south, text=paperprimary!80!black, yshift=3pt]
  at (axis cs:400, 84.1) {FP32};
\node[font=\small, anchor=south west, text=paperprimary!80!black, xshift=4pt, yshift=2pt]
  at (axis cs:600, 83.6) {FP16};
\node[font=\small, anchor=south west, text=paperprimary!80!black, xshift=4pt, yshift=2pt]
  at (axis cs:1800, 82.2) {INT8};
\node[font=\small, anchor=north east, text=paperprimary!80!black, xshift=-4pt, yshift=-2pt]
  at (axis cs:4000, 81) {W4A8};
\node[font=\small, anchor=north west, text=paperprimary!80!black, xshift=4pt, yshift=-2pt]
  at (axis cs:4800, 76.2) {W4A4};

\end{axis}
\end{tikzpicture}
\caption{Deployment design space for a representative model family on a single hardware target. Each marker represents a specific configuration of precision format, compiler, and batch size. Filled markers denote the five Pareto optimal configurations (non-dominated on both accuracy and throughput), connected by the frontier. Open markers represent the 19 dominated alternatives. Dashed lines mark minimum accuracy and throughput constraints. The shaded rectangle indicates the feasible deployment region satisfying both constraints.}
\label{fig:pareto}
\end{figure}

%% file: figures/fig_quant_effect.tex
% Figure 5: Effect of quantization on operational intensity and roofline position
\begin{figure}[t]
\centering
\begin{tikzpicture}
\begin{axis}[
  width=0.68\linewidth,
  height=0.40\linewidth,
  xlabel={Operational Intensity (FLOP/Byte)},
  ylabel={Attainable Perf.\ (TFLOP/s)},
  xmode=log, ymode=log,
  log basis x=2,
  xmin=1, xmax=512,
  ymin=1, ymax=400,
  grid=major,
  grid style={draw=papergrid!50, line width=0.3pt},
  axis line style={draw=paperaxis, thick},
  tick label style={font=\small},
  label style={font=\small}
]

% HBM bandwidth slope (2 TB/s)
\addplot[domain=1:512, samples=200, color=paperprimary, line width=1pt, solid]
  {min(2*x, 19.5)};

% FP16 TC ceiling
\addplot[domain=1:512, samples=200, color=papersecondary, line width=1pt, densely dashed]
  {min(2*x, 312)};

% Same kernel at three precisions. Starting kernel is near the FP32
% ridge (OI=8, memory-bound under FP32 ceiling). Quantization shifts it
% rightward AND upward past the FP32 ceiling into the FP16/INT8 TC regime.
%
% FP32: OI=8, perf ~65% of FP32 roof at that OI = 0.65*min(16,19.5) = 10.4
% FP16: OI=16, now above FP32 ceiling. Perf on FP16 TC slope = 0.65*2*16 = 20.8
% INT8: OI=32, well above FP32 ceiling. Perf on INT8 TC slope = 0.65*2*32 = 41.6

% FP32 kernel point (near FP32 ridge, memory-bound)
\addplot[only marks, mark=*, mark size=3.5pt, color=paperprimary, forget plot]
  coordinates {(8, 10.4)};
\node[font=\small, text=paperprimary, anchor=north, yshift=-5pt]
  at (axis cs:9.5, 10.4) {FP32};

% FP16 kernel point (above FP32 ceiling, on FP16 TC slope)
\addplot[only marks, mark=triangle*, mark size=3pt, color=papersecondary, forget plot]
  coordinates {(16, 20.8)};
\node[font=\small, text=papersecondary, anchor=west, xshift=4pt]
  at (axis cs:13, 13.8) {FP16};

% INT8 kernel point (well above FP32 ceiling, on INT8 TC slope)
\addplot[only marks, mark=square*, mark size=3pt, color=paperaccent, forget plot]
  coordinates {(32, 41.6)};
\node[font=\small, text=paperaccent, anchor=south west, xshift=3pt, yshift=2pt]
  at (axis cs:26, 19) {INT8};

% Arrow showing the full quantization shift FP32 -> INT8
\draw[-{Latex[length=2mm]}, thick, draw=paperaccent!70!black, densely dotted]
  (axis cs:8, 10.4) -- (axis cs:70, 90);
\node[font=\small, text=paperaccent!70!black, rotate=30, anchor=south, yshift=4pt]
  at (axis cs:120, 90) {Quantize};

\end{axis}
\end{tikzpicture}
\caption{Effect of precision reduction on roofline position. All three points represent the same kernel. FP32 executes on CUDA cores (blue roofline, 19.5 TFLOP/s ceiling). Quantizing to FP16 or INT8 promotes execution to Tensor Cores (red roofline, 312 TFLOP/s ceiling) while simultaneously increasing operational intensity by $2\text{--}4\times$ due to reduced bytes per element. The FP16 and INT8 points exceed the FP32 compute ceiling because they operate on a fundamentally higher performance tier of the hardware.}
\label{fig:quant_effect}
\end{figure}
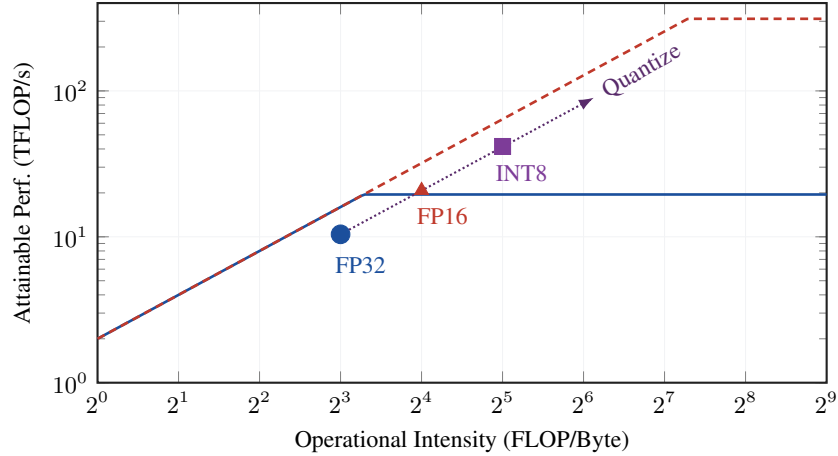

%% file: figures/fig_compiler_pipeline.tex
% Figure 4: Compiler transformation pipeline
\begin{figure*}[t]
\centering
\resizebox{\textwidth}{!}{%
\begin{tikzpicture}[
  >=Latex,
  font=\sffamily\small,
  pass/.style={
    rectangle, rounded corners=2pt, thick,
    minimum height=0.9cm, minimum width=2.0cm, align=center,
    draw=#1!70!black, fill=#1!10, text=paperblack,
    font=\sffamily\scriptsize
  },
  ir/.style={
    rectangle, rounded corners=1pt, thick,
    minimum height=0.7cm, minimum width=1.4cm, align=center,
    draw=papermuted!70!black, fill=paperbg, text=paperblack,
    font=\sffamily\scriptsize
  },
  arr/.style={-{Latex[length=2.0mm]}, line width=0.9pt, draw=paperaxis},
  annot/.style={font=\sffamily\scriptsize, text=papermuted, align=center}
]

\node[ir] (onnx) {ONNX /\\StableHLO};
\node[pass=paperprimary, right=0.5cm of onnx] (gopt) {Graph\\Optimization};
\node[ir, right=0.5cm of gopt] (fused) {Fused\\IR};
\node[pass=papertertiary, right=0.5cm of fused] (kgen) {Kernel\\Generation};
\node[ir, right=0.5cm of kgen] (tuned) {Tuned\\Kernels};
\node[pass=paperaccent, right=0.5cm of tuned] (lower) {Hardware\\Lowering};
\node[ir, right=0.5cm of lower] (exec) {Exec};

\draw[arr] (onnx) -- (gopt);
\draw[arr] (gopt) -- (fused);
\draw[arr] (fused) -- (kgen);
\draw[arr] (kgen) -- (tuned);
\draw[arr] (tuned) -- (lower);
\draw[arr] (lower) -- (exec);

\node[annot, above=2pt of gopt] {Fusion, folding,\\layout transform};
\node[annot, above=2pt of kgen] {Tiling, autotuning,\\vectorization};
\node[annot, above=2pt of lower] {PTX / LLVM IR,\\register alloc};

\draw[arr, densely dashed, draw=papermuted]
  (tuned.south) -- ++(0,-0.55) -| (kgen.south)
  node[pos=0.5, below=1pt, annot] {Autotuning feedback};

\end{tikzpicture}%
}
\caption{Compiler transformation pipeline. A high level graph IR passes through graph optimization (fusion, constant folding, layout transformation), kernel generation (tiling, autotuning, vectorization), and hardware lowering (PTX emission, register allocation). The dashed loop represents autotuning iterations.}
\label{fig:compiler_pipeline}
\end{figure*}
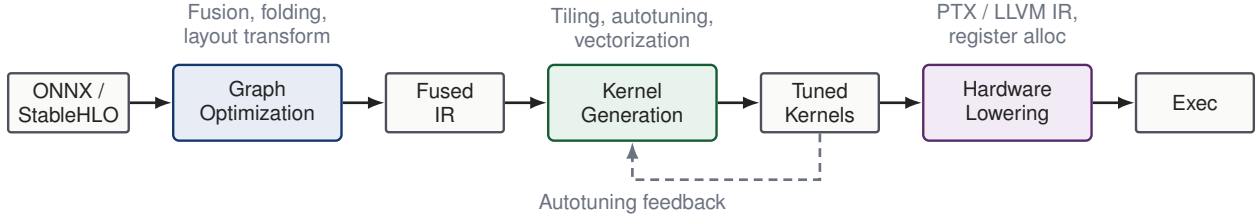

%% file: figures/fig_queuing.tex
% Figure 6: M/M/1 queuing amplification
% Shows response time normalized by service time as a function of utilization.
\begin{figure}[t]
\centering
\begin{tikzpicture}
\begin{axis}[
  width=0.70\linewidth,
  height=0.48\linewidth,
  xlabel={Utilization $\rho$},
  ylabel={$R / T_{\text{model}}$},
  xmin=0, xmax=1,
  ymin=0, ymax=25,
  grid=major,
  grid style={draw=papergrid!50, line width=0.3pt},
  axis line style={draw=paperaxis, thick},
  tick label style={font=\small},
  label style={font=\small},
  xtick={0, 0.2, 0.4, 0.6, 0.8, 1.0},
  ytick={0, 5, 10, 15, 20, 25}
]

% M/M/1 response time: R/T = 1/(1-rho)
\addplot[domain=0.01:0.97, samples=200, color=paperprimary, line width=1.2pt, solid]
  {1/(1-x)};

% Shaded danger zone (rho > 0.8)
\fill[papersecondary, opacity=0.08]
  (axis cs:0.8,0) rectangle (axis cs:1.0,25);
\node[font=\small\itshape, text=papersecondary!80!black, rotate=90, anchor=south]
  at (axis cs:0.97,4) {Saturation};

% Annotated points
\addplot[only marks, mark=*, mark size=2.5pt, color=paperprimary, forget plot]
  coordinates {(0.5, 2) (0.8, 5) (0.9, 10) (0.95, 20)};

% Labels for annotated points
\node[font=\small, text=paperprimary, anchor=south west, xshift=2pt]
  at (axis cs:0.46, 2.2) {$2\times$};
\node[font=\small, text=paperprimary, anchor=south east, xshift=2pt]
  at (axis cs:0.79, 5) {$5\times$};
\node[font=\small, text=paperprimary, anchor=west, xshift=4pt]
  at (axis cs:0.79, 10) {$10\times$};
\node[font=\small, text=paperprimary, anchor=west, xshift=4pt]
  at (axis cs:0.84, 20) {$20\times$};

\end{axis}
\end{tikzpicture}
\caption{Queuing amplification under the M/M/1 model. Response time normalized by model service time ($R / T_{\text{model}} = 1/(1-\rho)$) diverges nonlinearly as utilization $\rho$ approaches 1. At 50\% utilization, response time is $2\times$ the service time. At 90\%, it is $10\times$. At 95\%, it is $20\times$. The shaded region marks the saturation zone where small reductions in service time yield disproportionately large gains in end-to-end response time.}
\label{fig:queuing}
\end{figure}
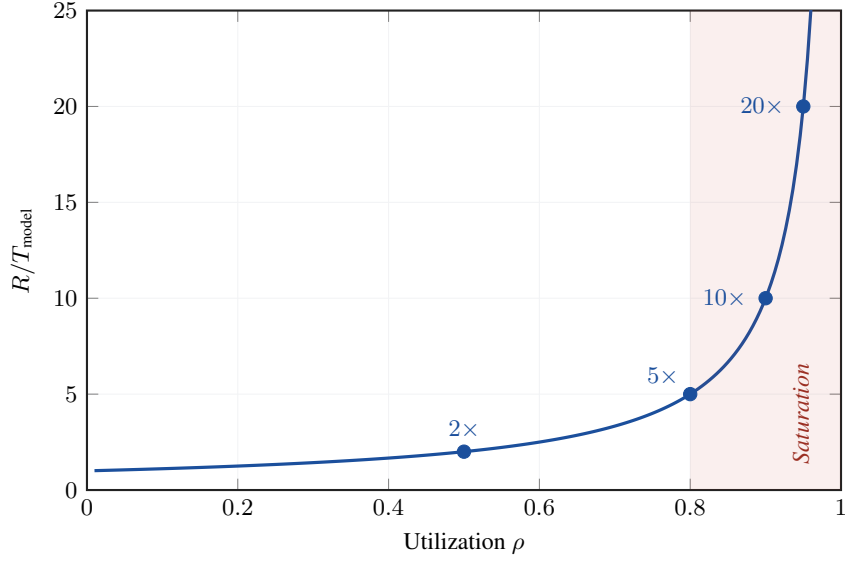

%% file: figures/fig_framework_latency.tex
% Figure 7: Framework latency comparison — Jetson AGX Orin
% Log-scale bar chart showing the 125x spread across runtimes for ResNet-152.
% Data from Ratul et al. (2025), Table 6.
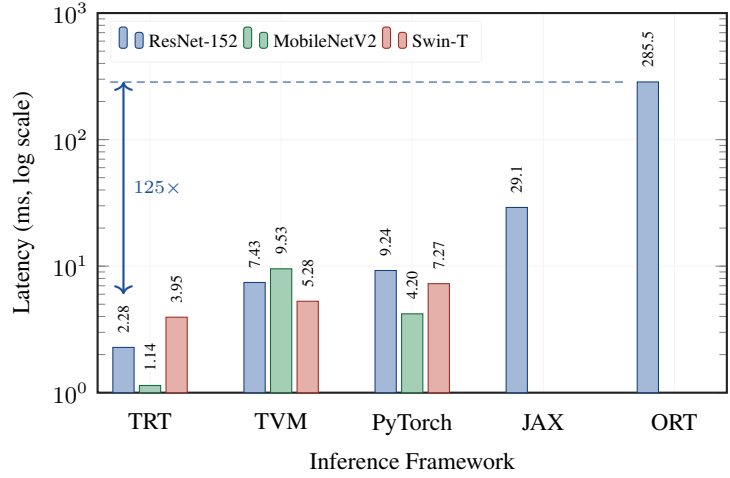
\begin{figure}[t]
\centering
\begin{tikzpicture}
\begin{axis}[
  width=0.6\linewidth,
  height=0.4\linewidth,
  ybar,
  bar width=8pt,
  ymode=log,
  ylabel={Latency (ms, log scale)},
  xlabel={Inference Framework},
  ymin=1, ymax=1000,
  symbolic x coords={TRT, TVM, PyTorch, JAX, ORT},
  xtick=data,
  xtick style={draw=none},
  x tick label style={font=\small},
  y tick label style={font=\small},
  label style={font=\small},
  grid=major,
  grid style={draw=papergrid!40, line width=0.2pt},
  axis line style={draw=paperaxis, thick},
  legend style={font=\scriptsize, at={(0.02,0.98)}, anchor=north west, draw=papergrid!60, fill=white, rounded corners=1pt},
  legend columns=3,
  nodes near coords,
  every node near coord/.append style={font=\tiny, rotate=90, anchor=west, xshift=2pt},
  point meta=explicit symbolic,
  clip=false
]

% ResNet-152
\addplot[fill=paperprimary!40, draw=paperprimary!80!black]
  coordinates {
    (TRT, 2.28) [2.28]
    (TVM, 7.43) [7.43]
    (PyTorch, 9.24) [9.24]
    (JAX, 29.1) [29.1]
    (ORT, 285.5) [285.5]
  };

% MobileNetV2
\addplot[fill=papertertiary!40, draw=papertertiary!80!black]
  coordinates {
    (TRT, 1.14) [1.14]
    (TVM, 9.53) [9.53]
    (PyTorch, 4.20) [4.20]
    (JAX, 0) []
    (ORT, 0) []
  };

% Swin-T
\addplot[fill=papersecondary!40, draw=papersecondary!80!black]
  coordinates {
    (TRT, 3.95) [3.95]
    (TVM, 5.28) [5.28]
    (PyTorch, 7.27) [7.27]
    (JAX, 0) []
    (ORT, 0) []
  };

\legend{ResNet-152, MobileNetV2, Swin-T}

\draw[paperprimary, densely dashed, thin]
  ([xshift=-15pt] axis cs:TRT, 285) -- ([xshift=30pt] axis cs:JAX, 285);

% Annotation: 125x spread
\draw[paperprimary, thick, <->]
  ([xshift=-10pt] axis cs:TRT, 6) -- node[right, font=\scriptsize, text=paperprimary, xshift=0pt] {$125\times$}
  ([xshift=-10pt] axis cs:TRT, 285.5);

\end{axis}
\end{tikzpicture}
\caption{Inference latency across five frameworks on Jetson AGX Orin (log scale). ResNet-152 exhibits a $125\times$ spread between TensorRT and ONNX Runtime. MobileNetV2 and Swin-T show progressively smaller spreads ($8.4\times$ and $1.9\times$ respectively), confirming that compiler optimization headroom is model dependent. Missing bars indicate frameworks not evaluated for that model. Data from Ratul et al.\ \cite{ratul2025electronics}.}
\label{fig:framework_latency}
\end{figure}

%% file: figures/fig_energy.tex
% Figure 8: Energy per inference — Jetson AGX Orin
% Grouped bar chart showing E = P * T.
% Highlights the Swin-T TVM vs TRT reversal with dimension-style annotation.
% Data derived from Table 6 (Ratul et al. 2025).
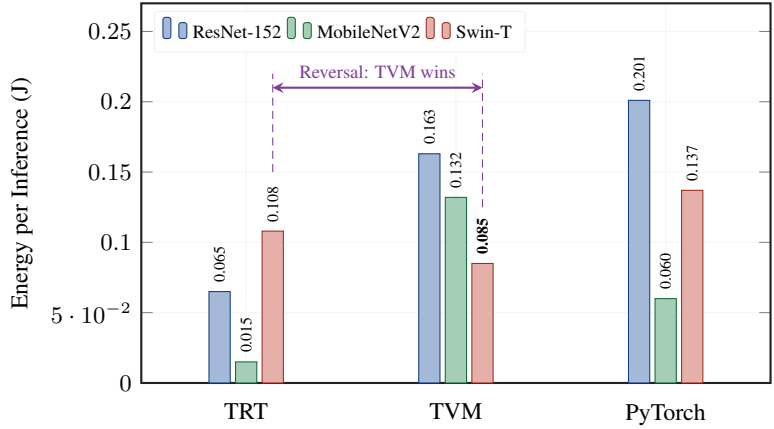
\begin{figure}[t]
\centering
\begin{tikzpicture}
\begin{axis}[
  width=0.6\linewidth,
  height=0.4\linewidth,
  ybar,
  bar width=8pt,
  ylabel={Energy per Inference (J)},
  ymin=0, ymax=0.27,
  symbolic x coords={TRT, TVM, PyTorch},
  xtick=data,
  xtick style={draw=none},
  x tick label style={font=\small},
  y tick label style={font=\small},
  label style={font=\small},
  grid=major,
  grid style={draw=papergrid!40, line width=0.2pt},
  axis line style={draw=paperaxis, thick},
  legend style={font=\scriptsize, at={(0.02,0.98)}, anchor=north west, draw=papergrid!60, fill=white, rounded corners=1pt},
  legend columns=3,
  nodes near coords,
  every node near coord/.append style={font=\tiny, rotate=90, anchor=west},
  point meta=explicit symbolic,
  enlarge x limits=0.25,
  clip=false
]

% ResNet-152 energy
\addplot[fill=paperprimary!40, draw=paperprimary!80!black]
  coordinates {
    (TRT, 0.065) [0.065]
    (TVM, 0.163) [0.163]
    (PyTorch, 0.201) [0.201]
  };

% MobileNetV2 energy
\addplot[fill=papertertiary!40, draw=papertertiary!80!black]
  coordinates {
    (TRT, 0.015) [0.015]
    (TVM, 0.132) [0.132]
    (PyTorch, 0.060) [0.060]
  };

% Swin-T energy
\addplot[fill=papersecondary!40, draw=papersecondary!80!black]
  coordinates {
    (TRT, 0.108) [0.108]
    (TVM, 0.085) [\textbf{0.085}]
    (PyTorch, 0.137) [0.137]
  };

\legend{ResNet-152, MobileNetV2, Swin-T}

% Dimension-style reversal annotation
% Vertical dashed lines from center of Swin-T bars at TRT and TVM
% Swin-T is the 3rd series, so offset xshift=+9pt at TRT, xshift=+9pt at TVM
\draw[paperaccent, densely dashed, thin]
  ([xshift=10pt] axis cs:TRT, 0.15) -- ([xshift=10pt] axis cs:TRT, 0.22);
\draw[paperaccent, densely dashed, thin]
  ([xshift=10pt] axis cs:TVM, 0.125) -- ([xshift=10pt] axis cs:TVM, 0.22);
% Horizontal dimension arrow between the two vertical lines
\draw[paperaccent, thick, stealth-stealth]
  ([xshift=10pt] axis cs:TRT, 0.21) -- 
  node[above, font=\scriptsize, text=paperaccent] {Reversal: TVM wins}
  ([xshift=10pt] axis cs:TVM, 0.21);

\end{axis}
\end{tikzpicture}
\caption{Derived energy per inference ($E = P \cdot T$) on Jetson AGX Orin across three model families and three frameworks. For ResNet-152 and MobileNetV2, TensorRT achieves the lowest energy despite the highest power draw, because its latency advantage dominates. For Swin-T, the ranking reverses (dimension annotation): TVM consumes less energy (0.085\,J) than TensorRT (0.108\,J) due to substantially lower power draw (16.0 vs 27.4\,W), demonstrating that energy-optimal and latency-optimal configurations can differ. Data derived from Ratul et al.\ \cite{ratul2025electronics}.}
\label{fig:energy}
\end{figure}

%% file: figures/fig_datacenter_frameworks.tex
% Figure 12: Data center LLM serving framework comparison — H100 SXM 80GB
% Dual-axis grouped bar chart: throughput (left) and TTFT p50 (right)
% Data from Spheron benchmarks \cite{spheron2026benchmark}.

\begin{figure}[t]
\centering
\begin{tikzpicture}
\begin{axis}[
  width=0.6\linewidth,
  height=0.4\linewidth,
  ybar,
  bar width=8pt,
  ylabel={Throughput (output tok/s)},
  xlabel={Concurrent Requests},
  ymin=0, ymax=3200,
  symbolic x coords={1, 10, 50, 100},
  xtick=data,
  xtick style={draw=none},
  x tick label style={font=\small},
  y tick label style={font=\small},
  label style={font=\small},
  grid=major,
  grid style={draw=papergrid!40, line width=0.2pt},
  axis line style={draw=paperaxis, thick},
  legend style={font=\scriptsize, at={(0.02,0.98)}, anchor=north west, draw=papergrid!60, fill=white, rounded corners=1pt},
  legend columns=3,
  nodes near coords,
  every node near coord/.append style={font=\tiny, rotate=90, anchor=west, xshift=0pt, yshift=1pt},
  point meta=explicit symbolic,
  enlarge x limits=0.18,
  clip=false
]

% vLLM throughput
\addplot[fill=paperprimary!40, draw=paperprimary!80!black]
  coordinates {
    (1, 120) [120]
    (10, 650) [650]
    (50, 1850) [1850]
    (100, 2400) [2400]
  };

% TensorRT-LLM throughput
\addplot[fill=papersecondary!40, draw=papersecondary!80!black]
  coordinates {
    (1, 130) [130]
    (10, 710) [710]
    (50, 2100) [2100]
    (100, 2780) [2780]
  };

% SGLang throughput
\addplot[fill=papertertiary!40, draw=papertertiary!80!black]
  coordinates {
    (1, 125) [125]
    (10, 680) [680]
    (50, 1920) [1920]
    (100, 2460) [2460]
  };

\legend{vLLM, TensorRT-LLM, SGLang}

\end{axis}
\end{tikzpicture}
\caption{LLM serving throughput on H100 SXM 80\,GB across three frameworks at increasing concurrency. TensorRT-LLM leads at all concurrency levels, but the gap compresses from $\sim$13\% at 50 concurrent requests to $\sim$12\% at 100, consistent with the queuing model prediction (Section~5.2) that system-level dynamics dominate compiler-level differences near saturation. Data from Spheron \cite{spheron2026benchmark}.}
\label{fig:datacenter_throughput}
\end{figure}
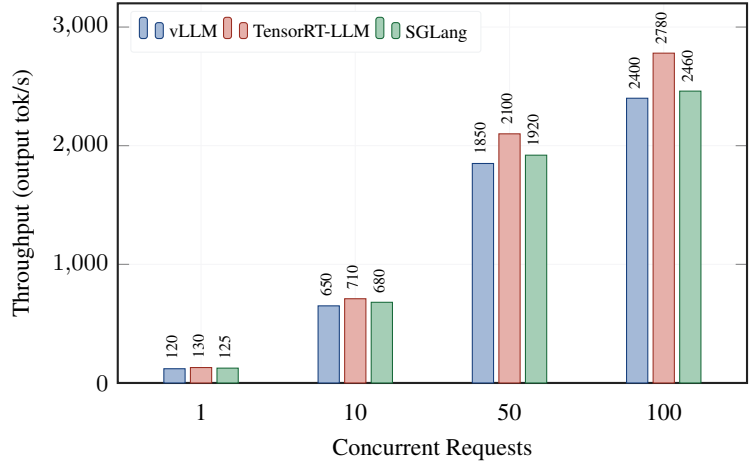

%% file: figures/fig_evidence_coverage.tex
% Figure 10: Evidence coverage radar chart
% Shows which metrics are well-reported vs underreported in the literature.
\begin{figure}[t]
\centering
\begin{tikzpicture}[scale=0.85]

% Precompute axis directions (7 axes, 360/7 ≈ 51.43 deg)
% Axis 1: Latency (90°), 2: Throughput (38.57°), 3: Accuracy (-12.86°),
% 4: Memory (-64.29°), 5: Power (-115.71°), 6: Energy (-167.14°), 7: Compilation (-218.57° = 141.43°)
\def\rscale{0.45}

% Grid rings
\foreach \r in {1,2,3,4,5} {
  \pgfmathsetmacro{\rad}{\r * \rscale}
  \draw[papergrid!40, very thin]
    (90:\rad) -- (38.57:\rad) -- (-12.86:\rad) -- (-64.29:\rad)
    -- (-115.71:\rad) -- (-167.14:\rad) -- (141.43:\rad) -- cycle;
}

% Axis spokes
\foreach \ang in {90, 38.57, -12.86, -64.29, -115.71, -167.14, 141.43} {
  \draw[papergrid, thin] (0,0) -- (\ang:5*\rscale);
}

% Axis labels
\node[font=\small, text=paperaxis] at (90:5*\rscale+0.65) {Latency};
\node[font=\small, text=paperaxis] at (38.57:5*\rscale+0.7) {Throughput};
\node[font=\small, text=paperaxis] at (-12.86:5*\rscale+0.8) {Accuracy};
\node[font=\small, text=paperaxis] at (-64.29:5*\rscale+0.5) {Memory};
\node[font=\small, text=paperaxis] at (-115.71:5*\rscale+0.4) {Power};
\node[font=\small, text=paperaxis] at (-167.14:5*\rscale+0.65) {Energy};
\node[font=\small, text=paperaxis] at (141.43:5*\rscale+0.7) {Compilation};

% Scale labels
\node[font=\tiny, text=papermuted, anchor=south east] at (90:1*\rscale) {1};
\node[font=\tiny, text=papermuted, anchor=south east] at (90:3*\rscale) {3};
\node[font=\tiny, text=papermuted, anchor=south east] at (90:5*\rscale) {5};

% Literature coverage polygon
% Latency=5, Throughput=4, Accuracy=5, Memory=3, Power=2, Energy=1, Compilation=1
\filldraw[paperprimary, fill=paperprimary, fill opacity=0.15, draw opacity=0.8, thick]
  (90:5*\rscale) -- (38.57:4*\rscale) -- (-12.86:5*\rscale) -- (-64.29:3*\rscale)
  -- (-115.71:2*\rscale) -- (-167.14:1*\rscale) -- (141.43:1*\rscale) -- cycle;

% Data points
\fill[paperprimary] (90:5*\rscale) circle (2pt);
\fill[paperprimary] (38.57:4*\rscale) circle (2pt);
\fill[paperprimary] (-12.86:5*\rscale) circle (2pt);
\fill[paperprimary] (-64.29:3*\rscale) circle (2pt);
\fill[paperprimary] (-115.71:2*\rscale) circle (2pt);
\fill[paperprimary] (-167.14:1*\rscale) circle (2pt);
\fill[paperprimary] (141.43:1*\rscale) circle (2pt);

% Ideal coverage polygon (all 5s, dashed red)
\draw[papersecondary, densely dashed, thick, opacity=0.5]
  (90:5*\rscale) -- (38.57:5*\rscale) -- (-12.86:5*\rscale) -- (-64.29:5*\rscale)
  -- (-115.71:5*\rscale) -- (-167.14:5*\rscale) -- (141.43:5*\rscale) -- cycle;

% Legend
\node[font=\scriptsize, text=paperprimary, anchor=west] at (2.6, -2.0) {\rule{10.5pt}{1pt}\ Reported in literature};
\node[font=\scriptsize, text=papersecondary, anchor=west] at (2.6, -2.35) {- - - Required for PEI ranking};

% Gap annotations on Energy and Compilation axes
\draw[paperaccent, thick, -stealth]
  (-167.14:1*\rscale+0.15) -- node[left, font=\tiny, text=paperaccent, yshift=-8pt, xshift=5pt] {Gap \ } (-167.14:5*\rscale-0.15);
\draw[paperaccent, thick, -stealth]
  (141.43:1*\rscale+0.15) -- node[left, font=\tiny, text=paperaccent] {Gap \ } (141.43:5*\rscale-0.15);

\end{tikzpicture}
\caption{Evidence coverage across the seven deployment metrics (Section~4). The blue polygon shows how frequently each metric is reported in the deployment literature (1 = rare, 5 = ubiquitous). The dashed red polygon indicates the coverage required for complete PEI ranking (Eq.~\ref{eq:pei}). Latency, throughput, and accuracy are well-reported. Power and memory are partially covered. Energy and compilation cost are systematically underreported, creating the gaps identified in Section~7.5.}
\label{fig:evidence_coverage}
\end{figure}
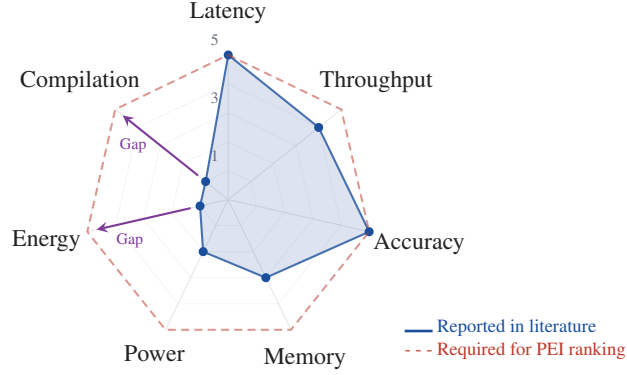

%% file: figures/fig_decision_flow.tex
% Figure 9: Deployment decision flowchart
% Staged measurement and scenario-based selection.
% Increased vertical spacing to make dashed stage separators visible.
\begin{figure}[!t]
\centering
\begin{tikzpicture}[
  node distance=1.0cm and 1.2cm,
  startstop/.style={rectangle, rounded corners, draw=paperaxis, thick, fill=paperprimary!8, minimum width=2.2cm, minimum height=0.7cm, font=\small, align=center},
  process/.style={rectangle, draw=paperaxis, thick, fill=paperprimary!20, minimum width=2.4cm, minimum height=0.7cm, font=\small, align=center},
  decision/.style={diamond, draw=paperaxis, thick, fill=orange!10, minimum width=1.6cm, minimum height=0.8cm, font=\small, align=center, aspect=2.2},
  arrow/.style={-stealth, thick, draw=paperaxis!80},
  lbl/.style={font=\scriptsize, text=papermuted}
]

% Start
\node[startstop] (start) {Candidate set $\mathcal{C}$};

% Stage 1
\node[decision, below=1.0cm of start] (mem) {$M(c) \le M_{\max}$?};
\draw[arrow] (start) -- (mem);
\node[startstop, right=1.5cm of mem, fill=papersecondary!20] (reject1) {Eliminate};
\draw[arrow] (mem) -- node[lbl, above] {No} (reject1);

% --- Dashed separator 1 (between Stage 1 and Stage 2) ---
\coordinate (sep1L) at ([xshift=-2.5cm, yshift=-0.35cm]mem.south);
\coordinate (sep1R) at ([xshift=1.5cm, yshift=-0.75cm]reject1.south);
\draw[black!50, densely dashed, line width=0.6pt] (sep1L) -- (sep1R);
\node[font=\scriptsize\bfseries, text=papermuted, anchor=east] at (sep1L) {Stage 1: Free};

% Stage 2
\node[decision, below=1.4cm of mem] (acc) {$A(c) \ge A_{\min}$?};
\draw[arrow] (mem) -- node[lbl, right] {Yes} (acc);
\node[startstop, right=1.5cm of acc, fill=papersecondary!20] (reject2) {Eliminate};
\draw[arrow] (acc) -- node[lbl, above] {No} (reject2);

% --- Dashed separator 2 (between Stage 2 and Stage 3) ---
\coordinate (sep2L) at ([xshift=-2.5cm, yshift=-0.35cm]acc.south);
\coordinate (sep2R) at ([xshift=1.5cm, yshift=-0.75cm]reject2.south);
\draw[black!50, densely dashed, line width=0.6pt] (sep2L) -- (sep2R);
\node[font=\scriptsize\bfseries, text=papermuted, anchor=east] at (sep2L) {Stage 2: Cheap};

% Stage 3: Measure
\node[process, below=1.4cm of acc] (measure) {Measure $T(c)$, $P(c)$\\on target hardware};
\draw[arrow] (acc) -- node[lbl, right] {Yes} (measure);

% Stage 3: Filter
\node[decision, below=1.0cm of measure] (lat) {$T(c) \le T_{\max}$?\\$E(c) \le E_{\max}$?};
\draw[arrow] (measure) -- (lat);
\node[startstop, right=1.5cm of lat, fill=papersecondary!20] (reject3) {Eliminate};
\draw[arrow] (lat) -- node[lbl, above] {No} (reject3);

% --- Dashed separator 3 (between Stage 3 and Ranking) ---
\coordinate (sep3L) at ([xshift=-2.5cm, yshift=-0.25cm]lat.south);
\coordinate (sep3R) at ([xshift=1.5cm, yshift=-0.85cm]reject3.south);
\draw[black!50, densely dashed, line width=0.6pt] (sep3L) -- (sep3R);
\node[font=\scriptsize\bfseries, text=papermuted, anchor=east] at (sep3L) {Stage 3: Expensive};

% Pareto + PEI
\node[process, below=1.4cm of lat] (pareto) {Pareto filter $\mathcal{P}$\\Rank by PEI};
\draw[arrow] (lat) -- node[lbl, right] {Yes} (pareto);

% Scenario split
\node[startstop, below left=0.8cm and 0.8cm of pareto, fill=papertertiary!25] (s1) {Latency-critical:\\min $T(c)$};
\node[startstop, below=0.8cm of pareto, fill=papertertiary!25] (s2) {Throughput-critical:\\max $Q(c)$ under SLA};
\node[startstop, below right=0.8cm and 0.8cm of pareto, fill=papertertiary!25] (s3) {Energy-critical:\\min $E(c)$};

\draw[arrow] (pareto) -- (s1);
\draw[arrow] (pareto) -- (s2);
\draw[arrow] (pareto) -- (s3);

\end{tikzpicture}
\caption{Deployment decision flowchart implementing the staged measurement protocol of Section~8.1 and the constraint-aware selection of Algorithm~\ref{alg:selection}. Candidates are eliminated by free constraints (memory), then cheap constraints (accuracy), before expensive on-device measurements (latency, power) are performed. Survivors are Pareto-filtered and ranked by PEI under the application-specific scenario. Dashed lines mark stage boundaries.}
\label{fig:decision_flow}
\end{figure}
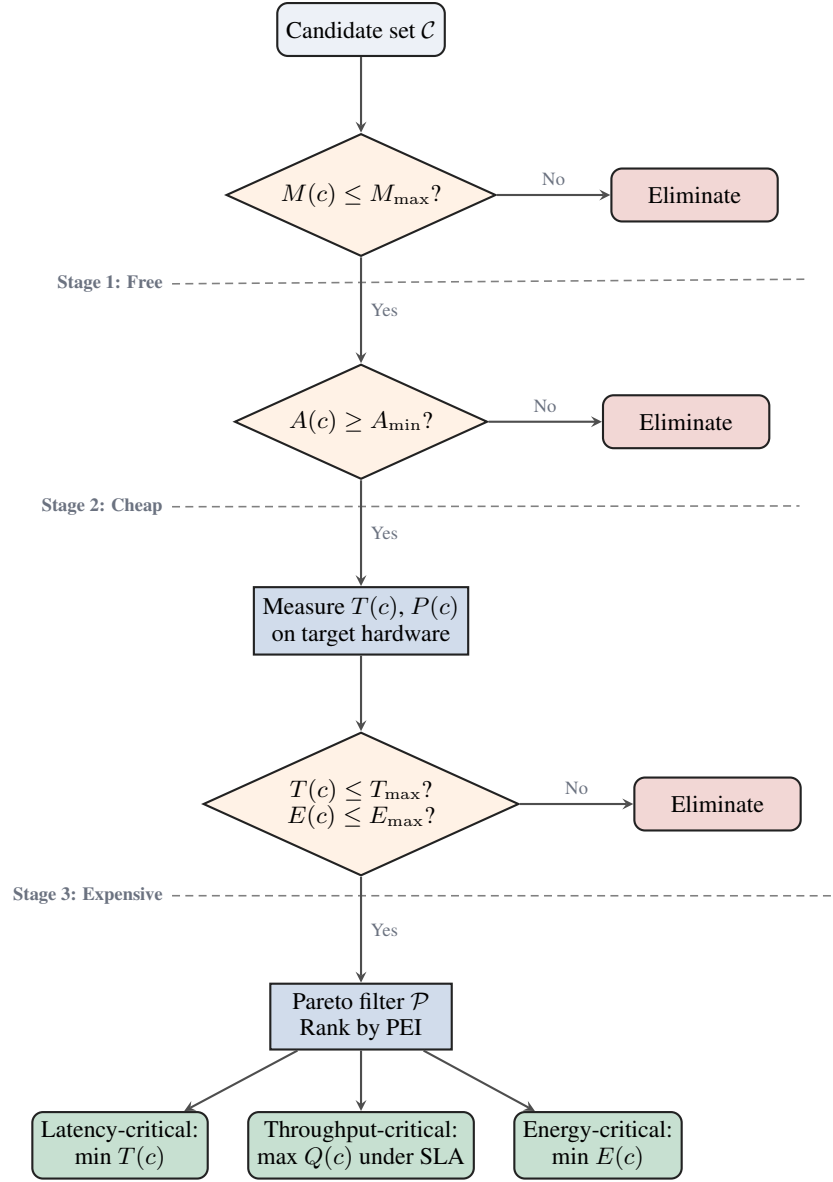

%% file: figures/fig_research_roadmap.tex
% Figure 11: Research roadmap
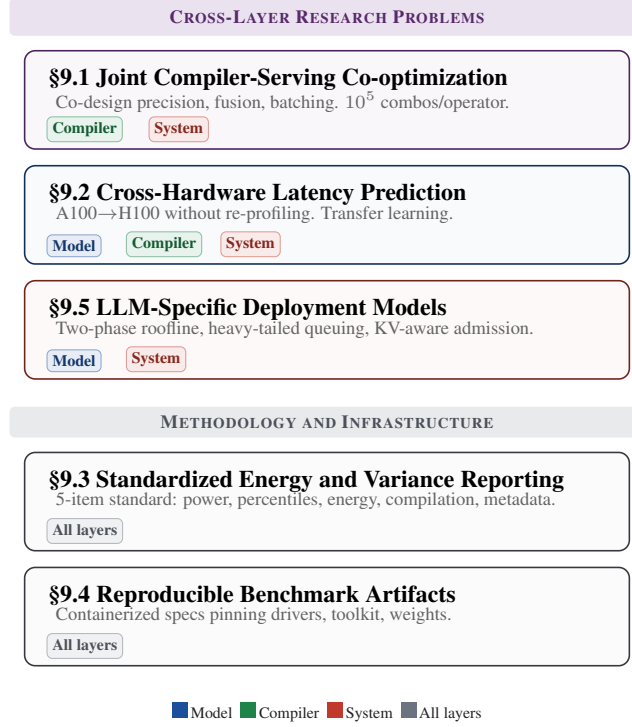
\begin{figure*}[t]
\centering
\begin{tikzpicture}[node distance=0.25cm]

% Category 1 header
\node[rectangle, rounded corners=2pt, draw=none, fill=paperaccent!15,
  minimum width=8.4cm, minimum height=0.5cm,
  font=\scriptsize\bfseries\scshape, text=paperaccent!70!black]
  (cat1) at (0, -0.12) {Cross-Layer Research Problems};

% === P1 ===
\node[rectangle, rounded corners=4pt, draw=paperaccent!60!black, line width=0.6pt,
  fill=paperaccent!4, minimum width=8.0cm, minimum height=1.3cm,
  inner sep=0pt, below=0.18cm of cat1] (p1) {};
\node[font=\small\bfseries, anchor=north west] at ([xshift=0.2cm, yshift=-0.12cm]p1.north west)
  {\S9.1 Joint Compiler-Serving Co-optimization};
\node[font=\scriptsize, text=paperaxis!70, anchor=north west] at ([xshift=0.3cm, yshift=-0.42cm]p1.north west)
  {Co-design precision, fusion, batching. $10^5$ combos/operator.};
\node[rectangle, rounded corners=2pt, draw=papertertiary!50, fill=papertertiary!10,
  font=\tiny\bfseries, text=papertertiary!70!black, inner sep=2pt,
  anchor=south west] at ([xshift=0.3cm, yshift=0.1cm]p1.south west) {Compiler};
\node[rectangle, rounded corners=2pt, draw=papersecondary!50, fill=papersecondary!10,
  font=\tiny\bfseries, text=papersecondary!70!black, inner sep=2pt,
  anchor=south west] at ([xshift=1.65cm, yshift=0.1cm]p1.south west) {System};

% === P2 ===
\node[rectangle, rounded corners=4pt, draw=paperprimary!60!black, line width=0.6pt,
  fill=paperprimary!3, minimum width=8.0cm, minimum height=1.3cm,
  inner sep=0pt, below=0.2cm of p1] (p2) {};
\node[font=\small\bfseries, anchor=north west] at ([xshift=0.2cm, yshift=-0.12cm]p2.north west)
  {\S9.2 Cross-Hardware Latency Prediction};
\node[font=\scriptsize, text=paperaxis!70, anchor=north west] at ([xshift=0.3cm, yshift=-0.42cm]p2.north west)
  {A100$\to$H100 without re-profiling. Transfer learning.};
\node[rectangle, rounded corners=2pt, draw=paperprimary!50, fill=paperprimary!10,
  font=\tiny\bfseries, text=paperprimary!70!black, inner sep=2pt,
  anchor=south west] at ([xshift=0.3cm, yshift=0.1cm]p2.south west) {Model};
\node[rectangle, rounded corners=2pt, draw=papertertiary!50, fill=papertertiary!10,
  font=\tiny\bfseries, text=papertertiary!70!black, inner sep=2pt,
  anchor=south west] at ([xshift=1.35cm, yshift=0.1cm]p2.south west) {Compiler};
\node[rectangle, rounded corners=2pt, draw=papersecondary!50, fill=papersecondary!10,
  font=\tiny\bfseries, text=papersecondary!70!black, inner sep=2pt,
  anchor=south west] at ([xshift=2.6cm, yshift=0.1cm]p2.south west) {System};

% === P5 ===
\node[rectangle, rounded corners=4pt, draw=papersecondary!60!black, line width=0.6pt,
  fill=papersecondary!3, minimum width=8.0cm, minimum height=1.3cm,
  inner sep=0pt, below=0.2cm of p2] (p5) {};
\node[font=\small\bfseries, anchor=north west] at ([xshift=0.2cm, yshift=-0.12cm]p5.north west)
  {\S9.5 LLM-Specific Deployment Models};
\node[font=\scriptsize, text=paperaxis!70, anchor=north west] at ([xshift=0.3cm, yshift=-0.42cm]p5.north west)
  {Two-phase roofline, heavy-tailed queuing, KV-aware admission.};
\node[rectangle, rounded corners=2pt, draw=paperprimary!50, fill=paperprimary!10,
  font=\tiny\bfseries, text=paperprimary!70!black, inner sep=2pt,
  anchor=south west] at ([xshift=0.3cm, yshift=0.1cm]p5.south west) {Model};
\node[rectangle, rounded corners=2pt, draw=papersecondary!50, fill=papersecondary!10,
  font=\tiny\bfseries, text=papersecondary!70!black, inner sep=2pt,
  anchor=south west] at ([xshift=1.35cm, yshift=0.1cm]p5.south west) {System};

% Category 2 header
\node[rectangle, rounded corners=2pt, draw=none, fill=papermuted!15,
  minimum width=8.4cm, minimum height=0.32cm,
  font=\scriptsize\bfseries\scshape, text=papermuted!70!black,
  below=0.35cm of p5] (cat2) {Methodology and Infrastructure};

% === P3 ===
\node[rectangle, rounded corners=4pt, draw=papermuted!50!black, line width=0.6pt,
  fill=papermuted!3, minimum width=8.0cm, minimum height=1.3cm,
  inner sep=0pt, below=0.18cm of cat2] (p3) {};
\node[font=\small\bfseries, anchor=north west] at ([xshift=0.2cm, yshift=-0.12cm]p3.north west)
  {\S9.3 Standardized Energy and Variance Reporting};
\node[font=\scriptsize, text=paperaxis!70, anchor=north west] at ([xshift=0.3cm, yshift=-0.42cm]p3.north west)
  {5-item standard: power, percentiles, energy, compilation, metadata.};
\node[rectangle, rounded corners=2pt, draw=papermuted!50, fill=papermuted!10,
  font=\tiny\bfseries, text=papermuted!70!black, inner sep=2pt,
  anchor=south west] at ([xshift=0.3cm, yshift=0.1cm]p3.south west) {All layers};

% === P4 ===
\node[rectangle, rounded corners=4pt, draw=papermuted!50!black, line width=0.6pt,
  fill=papermuted!3, minimum width=8.0cm, minimum height=1.3cm,
  inner sep=0pt, below=0.2cm of p3] (p4) {};
\node[font=\small\bfseries, anchor=north west] at ([xshift=0.2cm, yshift=-0.12cm]p4.north west)
  {\S9.4 Reproducible Benchmark Artifacts};
\node[font=\scriptsize, text=paperaxis!70, anchor=north west] at ([xshift=0.3cm, yshift=-0.42cm]p4.north west)
  {Containerized specs pinning drivers, toolkit, weights.};
\node[rectangle, rounded corners=2pt, draw=papermuted!50, fill=papermuted!10,
  font=\tiny\bfseries, text=papermuted!70!black, inner sep=2pt,
  anchor=south west] at ([xshift=0.3cm, yshift=0.1cm]p4.south west) {All layers};

% Legend
\node[font=\tiny, text=paperaxis, anchor=north] at ([yshift=-0.35cm]p4.south)
  {\textcolor{paperprimary}{\rule{6pt}{6pt}}\,Model\enspace
   \textcolor{papertertiary}{\rule{6pt}{6pt}}\,Compiler\enspace
   \textcolor{papersecondary}{\rule{6pt}{6pt}}\,System\enspace
   \textcolor{papermuted}{\rule{6pt}{6pt}}\,All layers};

\end{tikzpicture}
\caption{Research roadmap for the five open problems identified in Section~9. Each card describes the technical challenge and colored tags indicate affected deployment layers.}
\label{fig:research_roadmap}
\end{figure*}